\documentclass[12pt,journal,final,onecolumn,]{IEEEtran}
\makeatletter
\def\ps@headings{%
\def\@oddhead{\mbox{}\scriptsize\rightmark \hfil \thepage}%
\def\@evenhead{\scriptsize\thepage \hfil \leftmark\mbox{}}%
\def\@oddfoot{}%
\def\@evenfoot{}}

\usepackage{bbm}
\usepackage{amsfonts}
\usepackage[dvips]{graphicx}
\usepackage{times}
\usepackage{cite}
\usepackage{amsmath}
\usepackage{array}
\usepackage{amssymb}

\usepackage{stfloats}
\usepackage{slashbox}
\usepackage{graphicx}
\usepackage{footnote}
\usepackage{amsthm}
\usepackage{booktabs}
\usepackage{array}
\usepackage{algorithmic}
\usepackage{algorithm}
\usepackage{subeqnarray}
\usepackage{cases}
\usepackage{threeparttable}
\usepackage{color}
\usepackage{epstopdf}
\usepackage{bm}
\usepackage{epsfig}
\usepackage{epstopdf}

\usepackage{slashbox}

\usepackage{stmaryrd}
\usepackage{subfigure}
\usepackage{multirow}
\usepackage{subfig}
\usepackage{graphicx,times,amsmath} 
\DeclareMathOperator*{\argmin}{argmin}
\DeclareMathOperator*{\argmax}{argmax}

\newtheorem{Prob}{Problem}
\newtheorem{Lem}{Lemma}
\newtheorem{Thm}{Theorem}
\newtheorem{Cor}{Corollary}
\newtheorem{Rem}{Remark}

\begin{document}

\title{Modeling, Analysis, and Optimization of Coded Caching in Small-Cell Networks}
\author{
\IEEEauthorblockN{Xuejian Xu and Meixia Tao}\\
\thanks{This work has been accepted in part by IEEE ICC 2017 \cite{XuTao2017}.}
\thanks{The authors are with the Department of Electronic Engineering, Shanghai Jiao Tong University, Shanghai, China. Emails: \{manchester\_red, mxtao\}@sjtu.edu.cn. }
}

\maketitle

\begin{abstract}
Coded caching is able to exploit accumulated cache size and hence superior
to uncoded caching by distributing different fractions of a file in
different nodes.
This work investigates coded caching in a large-scale
small-cell network (SCN) where the locations of small base stations (SBSs)
are modeled by stochastic geometry. We first propose a content delivery
framework, where multiple SBSs that cache different coded packets of a
desired file transmit concurrently upon a user request and the user decodes
the signals using successive interference cancellation (SIC). We
characterize the performance of coded caching by two performance metrics,
\emph{average fractional offloaded traffic} (AFOT) and \emph{average ergodic
rate} (AER), for which a closed-form expression and a tractable expression
are derived, respectively, in the high signal-to-noise ratio region. We then
formulate the coded cache placement problem for AFOT maximization as a
multiple-choice knapsack problem (MCKP). By utilizing the analytical
properties of AFOT, a greedy but optimal algorithm is proposed. We also
consider the coded cache placement problem for AER maximization. By
converting this problem into a standard MCKP, a heuristic algorithm is
proposed. Analytical and numerical results reveal several design and
performance insights of coded caching in conjunction with SIC receiver in
interference-limited SCNs.

\end{abstract}

\section{Introduction}\label{sec:Introduction}
\subsection{Motivation and Related Works}
\label{subsec:MotivationandRelatedWorks}
The global mobile data traffic has undergone a fundamental shift from voices and messages to rich content  distributions, such as video streaming and application downloads. By prefetching popular contents during off-peak times at the edge of wireless networks, such as small base stations (SBSs), helper nodes, and user devices, wireless caching can alleviate peak-hour network congestion, provide traffic offloading, and improve users' quality of experience \cite{golrezaei2013,bastug2014,liu2016caching}.
It thus has attracted great attention recently.

Cache-enabled wireless networks operate in two phases in general, i.e., cache placement and content delivery \cite{Maddah-Ali2014}.
Cache placement is to place or update contents in each cache-enabled node subject to the storage size.
Content delivery is to deliver contents upon user requests conditioned on cache state as well as network channel condition.
While the design of the content delivery phase can be decoupled from the cache placement phase once the cache placement strategy is given, e.g., \cite{Tao2016,Zhou2016,Peng2014}, the design of the cache placement phase is tightly coupled with the content delivery phase in wireless networks, e.g. \cite{Malak2016,dai2016joint,blaszczyszyn2015optimal,liu2016caching2,choi2016,chen,chen2016}.
In this work, we are interested in the cache-enabled small-cell networks (SCNs), where each SBS is equipped with a local cache and can serve user requests based on its cached contents.
In a deterministic SCN with fixed connection topology, the cache placement problem is NP-hard \cite{FemtoCaching}.
Alternatively, the authors in \cite{bacstug2015cache,Yang2016analysis,blaszczyszyn2015optimal,liu2016caching2,choi2016,chen} apply tools from stochastic geometry to study the cache placement problem by assuming that the locations of SBSs follow a homogeneous Poisson point process (HPPP) in a large-scale SCN.
In particular, the authors in \cite{bacstug2015cache} analyze the outage probability and the average content delivery rate when each SBS stores the most popular contents and each user is only associated with the nearest SBS.
The authors in \cite{Yang2016analysis} analyze the average ergodic rate and the outage probability with the most popular contents caching in a three-tier heterogeneous network.
The works \cite{blaszczyszyn2015optimal,liu2016caching2,choi2016,chen} show that caching the contents randomly with the optimized probabilities has performance improvement, e.g., the cache hit probability, the successful offloading probability, and the average success probability of content
delivery, over caching the most popular contents since each user can be covered by multiple SBSs.

Recently, it is shown in \cite{chen2016,FemtoCaching,bioglio2015,Liao2016,altman2013,gu2013,Khreishah2016} that partitioning each file into multiple segments and caching different segments of a file in different SBSs can further enhance cache efficiency.
This is known as coded caching in \cite{FemtoCaching}.
When a user submits a file request, the file shall be delivered to the user from multiple nearby SBSs that have cached the different segments of the file, thereby exploiting the accumulated cache size.
The work \cite{chen2016} proposes a combined coded/uncoded caching strategy in disjoint cluster-centric SCN and analyzes the successful content delivery probability for a user located at the cluster center.
In the combined caching strategy, part of the cache space in each SBS is reserved for the most popular contents and the remaining is to cache different partitions of the less popular contents.
This strategy is, however, a heuristic one and cannot fully exploit the accumulated cache size.
In \cite{FemtoCaching,bioglio2015,Liao2016}, the maximum distance separable (MDS)-coded caching schemes are considered.
With an $(n, N)$ MDS code, each file is split into $n$ fragments and then encoded into $N\ (>n)$ coded packets. Any set of $n$ coded packets is sufficient to recover the file.
The works \cite{FemtoCaching} and \cite{bioglio2015} formulate the optimal MDS-coded cache placement problems to minimize the average file download time and the average backhaul rate, respectively.
Compared with \cite{bioglio2015}, the work \cite{Liao2016} considers a more practical scenario with heterogeneous file and cache sizes.
In \cite{altman2013,gu2013,Khreishah2016}, random linear network coding (RLNC) is applied to generate and store coded packets in SBSs.
Similar to MDS codes, when a file is split into $n$ fragments and the coding coefficients in RLNC are randomly selected from a large field, the file can be recovered from any $n$ coded packets with high probability.
All these works on MDS- or RLNC-coded caching in SCNs however assume an ideal error-free transmission with fixed connection topology in the content delivery phase.
Coded caching strategy in a realistic SCN with channel fading and inter-cell interferences remains uninvestigated.

\subsection{Contributions}\label{subsec:Contributions}
In this work, we apply stochastic geometry to model, analyze, and optimize coded caching in a cache-enabled SCN.
The locations of SBSs are assumed to follow an HPPP on a plane.
Each file is first partitioned into $n$ fragments and then encoded using either MDS codes or RLNC into an arbitrarily large number of coded packets, such that any set of $n$ coded packets is enough to recover the original file.
Each SBS stores a certain number of different coded packets for each file subject to a finite cache size.
When a user requests a file, the user will be associated with a sufficient number of SBSs that have cached the coded packets of the desired file for file downloading.
The main contributions of this work are summarized as follows.
\begin{itemize}
  \item \emph{A content delivery framework with SIC receiver}: We propose a new content delivery framework for coded caching, where multiple SBSs that cache the different coded packets of a desired file transmit concurrently upon a user request and the user decodes the signals successively using a successive interference cancellation (SIC)-based receiver. The number of concurrently transmitting SBSs, or equivalently, the number of SIC decoding layers, depends on the number of coded packets of the desired file cached in each SBS.
      We also obtain a closed-form expression to tightly approximate the success probability of each SIC decoding layer for a typical user in the high signal-to-noise ratio (SNR) region.
  \item \emph{New performance metrics and analysis}:
      We introduce and analyze two performance metrics for coded caching in SCNs.
      One is the average fractional offloaded traffic (AFOT).
      It measures the average fraction of each file that can be successfully delivered by the cache-enabled SBSs and hence offloaded from the core network at a given target SIC decoding threshold.
      Based on the success probability of each SIC decoding layer,
      we obtain a closed-form expression of AFOT.
      We show that the fractional offloaded traffic (FOT) for a given file is an increasing, concave, and piece-wise arithmetic sequence as the number of its coded packets cached in each SBS increases.
      The other is the average ergodic rate (AER), which characterizes the average ergodic data rate of a typical user achievable over the cache-enabled SBSs.
      We then obtain a tractable expression for AER.
  \item \emph{Optimization of coded cache placement}: We formulate two coded cache placement problems for AFOT maximization and AER maximization, respectively.
      The AFOT maximization problem is a multiple-choice knapsack problem (MCKP).
      Utilizing the analytical properties of AFOT, we propose a greedy-based algorithm that finds the global optimal solution efficiently.
      In the limiting case where each file can be split into infinite number of fragments, we also reformulate the discrete problem for AFOT maximization into a continuous one, which is shown to be convex.
      The optimal solution of the continuous problem provides a performance upper bound for the discrete problem.
      It also prompts us to propose a low-complexity algorithm with high performance for the discrete problem.
      To solve the coded cache placement problem for AER maximization, we transform it into a standard MCKP, for which a heuristic algorithm is proposed.
  \item \emph{Performance and design insights via numerical simulations}:
      Extensive numerical results demonstrate that it is good enough to split each file into not-so-large, e.g., $8$, fragments for coded caching to approach the performance upper bound of AFOT.
      Results also show that coded caching can well exploit the accumulated cache size among neighboring SBSs and hence significantly offload more traffic than uncoded caching when the decoding threshold for each SIC decoding layer is low.
      On the other hand, the gain of coded caching vanishes in terms of the AER performance unless the file popularity distribution is close to uniform.
\end{itemize}

\subsection{Paper Organization}\label{subsec:PaperOrganization}
This paper is organized as follows.
We present the network model, coded caching with SIC receiver, and the performance metrics, namely, AFOT and AER in Section~\ref{sec:SystemModel}.
We analyze the proposed performance metrics in Section~\ref{sec:PerformanceAnalysis} and formulate the coded cache problems into MCKP to optimize the caching vector in Section~\ref{sec:OptimizationofCodedCaching}.
Numerical results are presented in Section~\ref{sec:NumericalResults} and Section~\ref{sec:Conclusion} concludes the paper.

\section{System Model}\label{sec:SystemModel}
\subsection{Network Model}\label{subsec:netmodel}
We consider a cache-enabled SCN where the SBSs and users are spatially distributed on a plane according to two independent HPPPs with densities $\lambda_b$ and $\lambda_u$, respectively.
Each SBS is equipped with a local cache and serves mobile users with its cached contents.
All SBSs and users are equipped with single antenna.
Without loss of generality, we focus on the analysis at a typical user $u_0$.
We define $\phi_k$ as the $k$-th nearest SBS to $u_0$ with distance denoted by $r_k$.
The network is assumed to be fully loaded so that each SBS transmits signals all the time, which is the worst-case scenario for performance analysis.
There is also a tier of macro base stations (MBSs) in the network, which are connected to the core network via backhaul links and communicate with users only when the user demands cannot be satisfied by SBSs.
We assume that the MBSs and SBSs operate over non-overlapping frequency bands to avoid the inter-tier interference.

In the downlink channel, we consider both large-scale fading and small-scale fading. The large-scale fading is modeled by a standard distance-dependent power law pathloss attenuation.
For small-scale fading, Rayleigh fading is considered. The transmission power of each SBS is $P$. Hence, the received signal power at $u_0$ from $\phi_k$ can be written as $|h_k|^2 r_k^{-\alpha} P$, where $h_k\sim\mathcal {CN}(0, 1)$ is the complex-valued channel coefficient from $\phi_k$ to $u_0$ and $\alpha > 2$ is the pathloss exponent.

\subsection{Coded Caching Model}\label{subsect:CodedCachingModel}
Consider a content library consisting of $F$ files, denoted by $\mathcal{W}\triangleq\{W_1,W_2,\ldots,W_F\}$.
Each file has the same size of $B$ bits. The popularity of $W_j$ is denoted as $p_j$, for $j\in\mathcal{F}\triangleq\{1,2,\ldots,F\}$, satisfying $0\leq p_j\leq1$ and $\sum\limits_{j=1}^F p_j = 1$. Without loss of generality, we assume $p_1\geq p_2\geq \ldots \geq p_F$.

A coded caching scheme similar to \cite{bioglio2015} is adopted in this work. Each file in the content library is split into $n$ equal-sized fragments, for $n$ being any positive integer. The $n$ fragments of each file are then encoded into an arbitrarily large number of  packets.
Here, we mention \emph{fragment} as the original part of a file and \emph{packet} as the coded part of a file.
We do not restrict to any specific coding scheme and only require that the file can be decoded successfully from any $n$ coded packets.
In practice, this requirement can be satisfied by using MDS codes or satisfied with high probability by using RLNC.

Each SBS has a local cache and can store up to $M B$ bits $(M < F)$. In the cache placement phase, each SBS stores $m_j\in \mathcal{N}\triangleq\{0,1,\ldots, n\}$ different coded packets of $W_j$, for each $j\in \mathcal{F}$. We refer to $\mathbf{m} \triangleq [m_1,m_2,\ldots,m_F]$ as the \emph{caching vector}, which shall be optimized subject to the following cache size constraint:
\begin{align}
 &\sum\limits_{j=1}^F \frac{m_j}n \leq M. \label{eqn:discrete-constraint-1}
\end{align}

The specific packet storing approach depends on the coding scheme in the cache placement phase.
For example, if an $(n,N)$ MDS code with rate $\frac nN$ is used, where $n$ is the number of the fragments that each file is partitioned into and $N$ is the number of the coded packets based on the $n$ fragments, we can use graph coloring to store the coded packets in each SBS.
In specific, we first tessellate the considered network area into $\lceil\frac n{m_j}\rceil$-th order Voronoi regions, each of which is associated with $\lceil\frac n{m_j}\rceil$ SBSs so that any point in this region is closer to these $\lceil\frac n{m_j}\rceil$ SBSs than any other SBS in the network area.
Then, we construct a graph, where each SBS is represented by a vertex and every pair of SBSs that belong to the same $\lceil\frac n{m_j}\rceil$-th order Voronoi region are connected with an edge.
Let every distinct set of $m_j$ coded packets be represented by a color.
The maximum number of colors is $\bigl\lfloor\frac N{m_j}\bigr\rfloor$.
We then color the vertices of the graph so that no adjacent vertices share the same color, which is clearly feasible when $N$ is sufficiently large.
By this coloring-based packet storing approach, the $\lceil\frac n{m_j}\rceil$ SBSs in each Voronoi region are guaranteed to have different coded packets.
On the other hand, if RLNC is used by each SBS to encode the $n$ fragments with the coding coefficients independently and uniformly selected from a large field, the $m_j$ coded packets generated and stored in each SBS are ensured to be linearly independent with high probability.
In specific, let $\mathbb{F}_q$ denote the finite field where the random coding coefficients are chosen from and let $A$ denote the event that any $n$ coded packets collected from any $\lceil\frac n{m_j}\rceil$ SBSs in the network can recover the $n$ uncoded fragments and hence the complete file.
According to \cite{altman2013}, we have that $P(A) \geq \left(1-\frac 1q\right)^n$.
It is seen that when $q$ is large enough, $P(A)$ is very large.

\subsection{Content Request and Delivery Model}
\label{subsec:ContentRequestandDeliveryModel}
Each user requests for one file independently with probability according to the file popularity distribution.
Consider that the typical user $u_0$ sends a request for $W_j$, $j\in \mathcal F$. If $m_j=0$, i.e., no SBS has stored any part of $W_j$, $u_0$ will be served by an MBS, which is assumed to have access to all files through backhaul link.
 If $m_j\neq0$, $u_0$ will be associated with the set of $\bigl\lceil\frac{n}{m_j}\bigr\rceil$ nearest SBSs, denoted as $\Phi_j \triangleq \{\phi_1, \phi_2, \ldots, \phi_{\lceil\frac{n}{m_j}\rceil}\}$, to retrieve enough number of coded packets for the file decoding\footnote{The work can be extended to other user association strategies by considering different practical constraints.
 For instance, due to the limited SIC decoding capability, a user may only be associated with a maximum of $K (< n)$ nearest SBSs.
 Likewise, due the the large signal attenuation, a user may only be associated with the SBSs that are within a given distance.}.
The $\bigl\lceil\frac{n}{m_j}\bigr\rceil$ SBSs transmit their own cached packets concurrently over the same resource block.
Users adopt SIC-based receiver to decode the signals from multiple SBSs successively in the descending order of the average received signal strength as in \cite{chen2016}.
More specifically, the signal from the nearest SBS $\phi_1$ is decoded first and, if successful, subtracted from the received signal, then the signal from the second nearest SBS $\phi_2$ is decoded based on the residue received signal. The procedure repeats till the signals from all the SBSs in $\Phi_j$ are decoded.
The total number of decoding layers in the SIC-based receiver (equivalent to the total number of SBSs that transmit concurrently to $u_0$) is $\bigl\lceil\frac{n}{m_j}\bigr\rceil$.
This number varies between $1$ and $n$, depending on the different choices of $m_j$ for $W_j$.
Thus, the coding parameter $n$ for content caching sets the highest possible SIC receiver complexity.

Note that to decode the signal from $\phi_k$, the signals from $\{\phi_1, \phi_2, \ldots, \phi_{k-1}\}$ must be decoded and subtracted first from the received signal, and then only the signals from $\phi_{k+1}$ and beyond are seen as interferences.
Define $\Phi_{O}^k \triangleq \{\phi_{i}|r_{i}\geq r_k\} \setminus \{\phi_k\}$ as the set of SBSs whose distances to $u_0$ are larger than or equal to $r_k$ , for $k\in \mathcal{N^+}\triangleq\mathcal{N}\setminus \{0\}$.
Then, the residue received signal for decoding the signal from $\phi_k$ is given by
  \begin{align}\label{}
    y_k = h_kr_k^{-\frac{\alpha}2}\sqrt{P} s_k + \sum\limits_{i \in \Phi_{O}^k} h_ir_i^{-\frac{\alpha}2} \sqrt{P}s_i + n_0, \ k\in \mathcal{N}^+,
 \end{align}
 where $s_k$ denotes the transmitted symbol from $\phi_k$ and is assumed independent for different $k$'s, and $n_0\sim\mathcal {CN}(0, N_0)$ is the complex additive white Gaussian noise of power $N_0$.

 We focus on the interference-limited network where the noise can be neglected. The signal-to-interference ratio (SIR) for decoding $s_k$ is given by
\begin{align}\label{eqn:SIR-C}
    \text{SIR}_k = \frac{|h_k|^2r_k^{-\alpha}}{\sum\limits_{i \in \Phi_{O}^k} |h_i|^2r_i^{-\alpha}},  \ k\in \mathcal{N}^+.
\end{align}

\subsection{Performance Metrics}
\label{subsec:PerformanceMetric}
In our work, we shall measure the performance of coded caching using two metrics, namely, the \emph{average fractional offloaded traffic (AFOT)} and the \emph{average ergodic rate (AER)}.
The former characterizes the average fraction of each file that can be successfully delivered by the cache-enabled SBSs and hence offloaded from the core network at a given target SIC decoding threshold.
The latter characterizes the average ergodic data rate of the typical user achievable over the cache-enabled SBSs.

\subsubsection{Average Fractional Offloaded Traffic (AFOT)}
\label{subsubsec:AverageFractionalOffloadedTraffic}
We first define the successful transmission probability of each decoding layer at a common SIR threshold $\tau$.
The signal from $\phi_k$ can only be decoded successfully if $\text{SIR}_k \geq \tau$ after the signals from $\{\phi_1, \phi_2, \ldots, \phi_{k-1}\}$ are all decoded and subtracted successfully.
The probability of this event is referred to as the \emph{success probability of the $k$-th decoding layer}, denoted by
\begin{align}\label{eqn:definition-q-k}
 q_k(\tau) \triangleq \text{Pr}\left[\text{SIR}_k\geq \tau \bigg| \bigcap_{i=1,2,\ldots,k-1}\text{SIR}_i\geq \tau\right],  \ k\in \mathcal{N}^+.
\end{align}
Otherwise, if $\text{SIR}_k < \tau$, the decoding of $s_k$ fails and the SIC process terminates.
The remaining coded packets will be acquired from an MBS for $u_0$ to recover the complete requested file.

Consider that $u_0$ sends a request for $W_j$ and $m_j$ coded packets of the file are stored in each SBS.
As mentioned above, the user can only decode the $m_j$ packets from $\phi_k$ after all the transmissions from the first $k-1$ nearest SBSs are successful. Once the decoding of the $k$-th layer fails, the SIC-based decoding process terminates.
Then, only $(k-1)m_j$ coded packets of $W_j$ can be offloaded by the SBSs. If all the $\bigl\lceil\frac{n}{m_j}\bigr\rceil$ decoding layers are successful, the entire file can be offloaded.
Therefore, the \emph{fractional offloaded traffic (FOT)} for $W_j$, denoted as $L[m_j]$, is given by
\begin{align}\label{eqn:R-mj}
 L[m_j] \triangleq
\begin{cases}
0, &m_j=0, \\
\sum\limits_{k=1}^{\bigl\lceil\frac{n}{m_j}\bigr\rceil} \frac{(k-1)m_j}{n}
\text{Pr} \left[\bigcap\limits_{i=1,2,\ldots,k-1}\text{SIR}_i\geq \tau, \text{SIR}_k< \tau\right] \\
\hspace{0.5cm}+ \text{Pr} \Bigg[\bigcap\limits_{i=1,2,\ldots,\lceil\frac{n}{m_j}\rceil}\text{SIR}_i\geq \tau\Bigg],  &m_j \in \mathcal{N^+}.
\end{cases}
\end{align}

By averaging over all files in $\mathcal{W}$, the AFOT, denoted as $L$, is given by
\begin{align}\label{eqn:theaveragebackhaulrateofcodedcachingscheme}
 L \triangleq \sum\limits_{j=1}^F p_j L[m_j].
\end{align}

\begin{Rem}
We would like to highlight that the decoding threshold $\tau$ in \eqref{eqn:definition-q-k} is set as a constant, independent of the decoding layers $\bigl\lceil\frac{n}{m_j}\bigr\rceil$ in the SIC receiver.
At first glance, this means that it only takes $\frac{m_j}n T_0$ seconds to deliver the requested file $W_j$ (no matter it succeeds or not) since each SBS only transmits $\frac{m_j}n B$ bits, where $T_0 = \frac B {W\log_2(1+\tau)}$ is the target delivery time of a file of $B$ bits in the traditional uncoded caching scheme with channel bandwidth of $W$ Hz.
However, $\bigl\lceil\frac{n}{m_j}\bigr\rceil$ SBSs are involved in delivering $W_j$. The accumulated time resource consumed by delivering $W_j$ thus becomes  $\bigl\lceil\frac{n}{m_j}\bigr\rceil \frac{m_j}{n} T_0$, which is equal to or slightly larger than $T_0$. Thus, having a constant $\tau$ in \eqref{eqn:definition-q-k} ensures a fair comparison between coded and uncoded caching in terms of the total network resource consumption.
\end{Rem}

\subsubsection{Average Ergodic Rate (AER)}
\label{subsubsec:AverageErgodicRate}
Again, consider that $u_0$ sends a request for $W_j$ and $m_j$ coded packets of the file are stored in each SBS.
Here, we assume that each SBS adopts a rate-adaptive transmission strategy so that the signals transmitted from all the SBSs in $\Phi_j$ are always decodable using the SIC receiver when $m_j \in \mathcal{N^+}$.
Since each SBS in $\Phi_j$ needs to transmit the same amount of $\frac{m_j}{n}B$ information bits
to $u_0$, the actual transmission rate of each SBS is chosen according to the worst SIR among all decoding layers in~\eqref{eqn:SIR-C}.
Thus, the \emph{ergodic rate} achievable by $u_0$ when requesting $W_j$, denoted as $R[m_j]$ bits/s/Hz, is given by
\begin{align}\label{eqn:definition-ergodic-rate}
  R[m_j] \triangleq
\begin{cases}
0, &m_j=0, \\
\bigl\lceil\frac{n}{m_j}\bigr\rceil \mathbb{E}\left[\log\left(1+\min\left(\text{SIR}_1,\ldots, \text{SIR}_{\lceil\frac{n}{m_j}\rceil}\right)\right)\right], &m_j \in \mathcal{N^+}.
\end{cases}
\end{align}


By averaging over all files in $\mathcal{W}$, the AER of the typical user achievable over the cache-enabled SBSs, denoted as $R$ bits/s/Hz, is given by
\begin{align}\label{eqn:definition-average-ergodic-rate}
  R \triangleq \sum\limits_{j=1}^F p_j R[m_j].
\end{align}


\section{Performance Analysis}
\label{sec:PerformanceAnalysis}
In this section, we shall analyze the performance of coded caching using the AFOT and AER defined in the previous section.
Especially, exploring the structural properties of FOT will benefit the coded caching optimization problem for AFOT maximization in Section~\ref{sec:OptimizationofCodedCaching}.

\subsection{Analysis of Fractional Offloaded Traffic $L[m_j]$}
\label{subsec:AnalysisandPropertiesofFOT}
We first obtain a closed-form expression for the success probability of the $k$-th decoding layer $q_k(\tau)$, defined in \eqref{eqn:definition-q-k}, then apply it to obtain a closed-form expression for $L[m_j]$ defined in \eqref{eqn:R-mj}.
 \begin{Lem}\label{Lem:CodedTransmissionCoverageProbability}
   The success probability of the $k$-th decoding layer with SIC receiver can be approximated as
 \begin{align}\label{eqn:CodedTransmissionCoverageProbability}
   q_k(\tau) = \left(1+\frac2{\alpha}\tau^{\frac2{\alpha}}B'\left(\frac2{\alpha},1-\frac2{\alpha},\frac1{1+\tau}\right)\right)^{-k}, \ k\in \mathcal{N}^+,
 \end{align}
 where $B'(x,y,z) \triangleq \int_z^1 u^{x-1} (1-u)^{y-1} \mathrm{d}u$ is the complementary incomplete Beta function.
 \end{Lem}
\begin{proof}Please see Appendix A.
\end{proof}

We shall demonstrate via simulation in Section~\ref{subsec:SuccessfulTransmissionProbabilityValidation} that the closed-form expression in Lemma~\ref{Lem:CodedTransmissionCoverageProbability} is a tight approximation of the true successful transmission probability of SBS $\phi_k$ after the signals from $\{\phi_1,\phi_2,\ldots,\phi_{k-1}\}$ are all successfully decoded and subtracted.

With Lemma~\ref{Lem:CodedTransmissionCoverageProbability} and definition in \eqref{eqn:R-mj}, FOT of $W_j$ can be obtained in the following lemma.
\begin{Lem}
\label{Lem:BRofCodedCaching}
 The fractional offloaded traffic of $W_j$ with caching parameter $m_j$ is given by
\begin{align}\label{eqn:R_c}
L[m_j] =
\begin{cases}
0, &m_j=0, \\
\frac{m_j}n\sum\limits_{k=1}^{\bigl\lceil\frac{n}{m_j}\bigr\rceil}C_k(\tau)
+\left(1-\frac{m_j}{n}\bigl\lceil\frac{n}{m_j}\bigr\rceil\right)
 C_{\lceil\frac{n}{m_j}\rceil}(\tau),  &m_j \in \mathcal{N^+},
\end{cases}
\end{align}
where
\begin{align}\label{eqn:c_k_tau}
   C_k(\tau) = \left(1+\frac2{\alpha}\tau^{\frac2{\alpha}} B'\left(\frac2{\alpha},1-\frac2{\alpha},\frac1{1+\tau}\right)\right)^{-\frac{k(k+1)}2}.
\end{align}
\end{Lem}
\begin{proof}
Please see Appendix B.
\end{proof}

By substituting \eqref{eqn:R_c} into \eqref{eqn:theaveragebackhaulrateofcodedcachingscheme}, a closed-form expression of AFOT for any given caching vector $\mathbf{m}$ is readily obtained.

Now we analyze the properties of $L[m_j]$, which will be essential for cache placement optimization in the next section.


\begin{Lem}\label{Lem:R-mj-piecewise}
$L[m_j]$ is a piece-wise arithmetic sequence, for $m_j \in\mathcal{N}$.
\end{Lem}
\begin{proof}
Define an integer set $\mathcal{M}_t \triangleq\{m|\frac{n}{t}\leq m< \frac{n}{t-1},m\in \mathbb{N^+}\}, \forall t\in\mathcal{T}_n \triangleq \{t|t=\bigl\lceil\frac{n}{m_j}\bigr\rceil, \forall m_j \in \mathcal{N^+}\}\setminus \{1\}$ and $\mathcal{M}_1 \triangleq \{n\}$.
 For example, $\mathcal{M}_2=\{4,5,6,7\}$ and $\mathcal{M}_3=\{3\}$ when $n=8$.
 Then, we have $\bigl\lceil\frac{n}{m_j}\bigr\rceil = t$, $\forall m_j \in \mathcal{M}_t$.
Thus, for all $m_j$'s such that $\{m_j,m_j-1\}\subseteq\mathcal{M}_t$, the first difference of $L[m_j]$ is given by
\begin{align}\label{eqn:difference-R-c}
 d_t &\triangleq L[m_j]-L[m_j-1] = \frac 1n\left(\sum\limits_{k=1}^t C_k(\tau) - tC_t(\tau)\right).
\end{align}
Equation \eqref{eqn:difference-R-c} indicates that when multiple $m_j$'s all belong to $\mathcal{M}_t$, the difference of $L[m_j]$, i.e., $d_t$, is same and thus $L[m_j]$ is an arithmetic sequence.
If $m_j \in \mathcal{M}_t$ but $m_j-1\in \mathcal{M}_{t'}$, for $t' \ne t$, then we denote the first difference as $d_{t,t'} \triangleq L[m_j] - L[m_j-1]$. Thus, $L[m_j]$ is a piece-wise arithmetic sequence with first difference given by $d_t$ when $\{m_j, m_j-1\} \in \mathcal{M}_t$ or $d_{t,t'}$ when $m_j \in \mathcal{M}_t$ and $m_j-1 \in \mathcal{M}_{t'}$. The Lemma is thus proven.
\end{proof}


\begin{table}
\centering
\caption{Offloaded traffic difference table at $n=8$.}\label{Table:1}
\begin{tabular}{|c|c|c|c|c|c|c|c|c|c|}
\hline
$m_j$ & 0 & 1 & 2 & 3 & 4 & 5 & 6 & 7 & 8 \\
\hline
$\bigl\lceil\frac{n}{m_j}\bigr\rceil$ & $\backslash$ & 8 & 4 & 3 & \multicolumn{4}{c|}{2} & 1 \\
\hline
$\delta_j$ & $\backslash$ & $d_{8}$ & $d_{4,8}$ & $d_{3,4}$ & $d_{2,3}$ & \multicolumn{4}{c|}{$d_2$} \\
\hline
\end{tabular}
\vspace{0.08in}
\end{table}

By Lemma~\ref{Lem:R-mj-piecewise} and its proof, we can construct a difference table of $L[m_j]$, for $m_j \in\mathcal{N}$, termed as the \emph{offloaded traffic difference table}.
An example with $n=8$ is shown in TABLE~\ref{Table:1}.
Here, $\delta_j = L[m_j] - L[m_j - 1]$.
For any given $n$, an offloaded traffic difference table can be constructed.
We shall use this table to design a greedy-based optimal algorithm and a low-complexity algorithm with high performance for the cache placement problem for AFOT maximization in the next section.
The total number of distinct $\delta_j$'s for general $n$, denoted as $N_n$, depends on the size of set $\mathcal{T}_n$, i.e., the number of piece-wise regions for $m_j$, which is usually much smaller than $n$.
For example,  $N_8 = 5$, $N_{16} = 8$ and $N_{32} = 12$. When $n\to\infty$, $N_n\to 2\bigl\lceil\sqrt n\bigr\rceil$.
The proof is skipped due to the page limit.

Now, we consider the limiting case where $n\to \infty$. That means the file size $B$ is very large so that each file can be split into infinite number of fragments.
Define a new variable $x_j=\frac {m_j}n$. When $n\to\infty$, $x_j$ can be relaxed as a continuous variable within $[0,1]$. Then, the discrete function $L[m_j]$ in \eqref{eqn:R_c} can be relaxed as a continuous function $L(x_j)$:
\begin{align}\label{}
  L(x_j) = x_j\sum\limits_{k=1}^{\bigl\lceil\frac1{x_j}\bigr\rceil} C_k(\tau) + \left(1-x_j\biggl\lceil\frac1{x_j}\biggr\rceil\right)C_{\bigl\lceil\frac1{x_j}\bigr\rceil}(\tau).
\end{align}

The properties of $L(x_j)$ are summarized in the following lemma.
\begin{Lem}\label{Lem:R-xj-increase-concave}
 $L(x_j)$ is an increasing, concave and piece-wise linear function of $x_j$, for $0\leq x_j\leq 1$.
\end{Lem}
\begin{proof}Please see Appendix C.
\end{proof}

For any given $n$, $L[m_j]$ can be viewed as the sampled function of $L(x_j)$ at sampling instance $x_j=\frac{m_j}n$.
Since $L(x_j)$ is increasing and concave by Lemma~\ref{Lem:R-xj-increase-concave}, $L[m_j]$ is also increasing and concave.
We thus have the following corollary.
\begin{Cor}
\label{Cor:Delta-mj-decreasing}
 $L[m_j]$ is an increasing and concave sequence, i.e., $0<L[m_j]-L[m_j-1]\leq L[m_j-1]-L[m_j-2]$, for $m_j \in\{2,3,\ldots,n\}$.
\end{Cor}

\begin{figure}[tb]
\begin{center}
  \includegraphics[width=9.5cm]{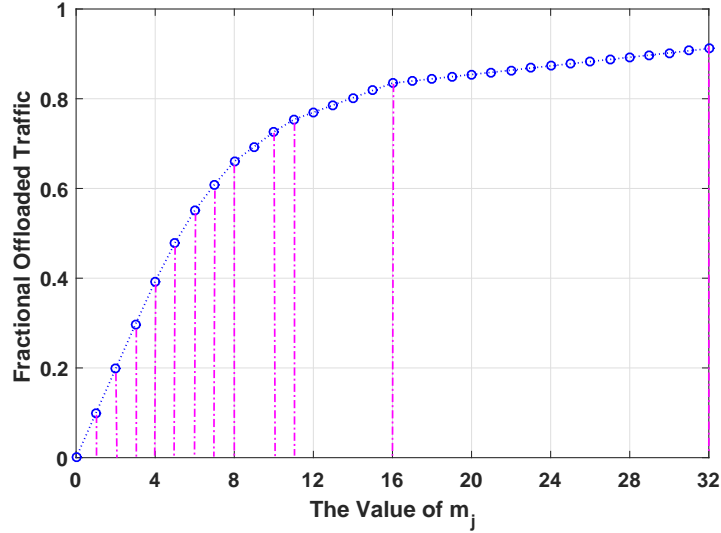}\\
  \caption{Fractional offloaded traffic $L[m_j]$ with $\alpha=4$, $\tau=0.1$, $n=32$.}
  \label{Fig:linear-discrete}
\end{center}
\end{figure}

Fig.~\ref{Fig:linear-discrete} depicts a numerical example of $L[m_j]$ when $n=32$.
Notice that $d_{t,t+1}-d_{t+1} = \left(1-t\frac{m_j}n\right)(C_t(\tau)-C_{t+1}(\tau))$ when $t$ and $t+1$ both belong to $\mathcal T_{n}$. Thus, $d_{1,2}=d_2$ and $d_{2,3}=d_3$ when $n=32$.
Fig.~\ref{Fig:linear-discrete} clearly demonstrates the properties in Lemma~\ref{Lem:R-mj-piecewise} and Corollary~\ref{Cor:Delta-mj-decreasing}.

\subsection{Analysis of Ergodic Rate $R[m_j]$} \label{subsec:AnalysisoftheErgodicRate}
By using the definition in \eqref{eqn:definition-ergodic-rate}, the ergodic rate $R[m_j]$ of $u_0$ when requesting $W_j$ can be expressed in a tractable form below.
\begin{Lem}
\label{Lem:ErgodicRateofCodedCaching}
 The ergodic rate for $W_j$ with caching parameter $m_j$ is given by
\begin{align}\label{eqn:Ergodic-Rate}
R[m_j] =
\begin{cases}
0, &m_j=0, \\
\bigl\lceil\frac{n}{m_j}\bigr\rceil \int_0^{\infty} C_{\lceil\frac{n}{m_j}\rceil}\left(2^r-1\right) \mathrm{d}r,  &m_j \in \mathcal{N^+},
\end{cases}
\end{align}
where $C_{\lceil\frac{n}{m_j}\rceil}\left(2^r-1\right)$ is given in~\eqref{eqn:c_k_tau} by letting $k=\lceil\frac{n}{m_j}\rceil$ and $\tau=2^r-1$.
\end{Lem}
\begin{proof}
Please see Appendix D.
\end{proof}


By substituting \eqref{eqn:Ergodic-Rate} into \eqref{eqn:definition-average-ergodic-rate}, the expression of AER for any given caching vector $\mathbf{m}$ can be obtained.
Unlike the expression of FOT $L[m_j]$ in Lemma~\ref{Lem:BRofCodedCaching}, which is in closed-form, the expression of ergodic rate $R[m_j]$ involves integral, though still tractable.
It is thus difficult to analyze its properties.

\section{Optimization of Coded Caching}
\label{sec:OptimizationofCodedCaching}

In this section, we would like to maximize the AFOT and AER, respectively,  by optimizing the caching vector $\mathbf{m}$.
In the AFOT optimization, we first study the discrete optimization problem when the coding parameter $n$ is finite, and then study the continuous optimization problem when $n \to \infty$.
In the AER optimization, we only consider the discrete case when $n$ is finite.

\subsection{Discrete Problem for AFOT Optimization}
\label{subsec:DiscreteOptimizationonMaximizingAFOT}
The objective is to maximize the AFOT in~\eqref{eqn:theaveragebackhaulrateofcodedcachingscheme} by optimizing the caching vector $\mathbf m$ subject to the cache size constraint~\eqref{eqn:discrete-constraint-1} at a given finite $n$. This is formulated as follows.
\begin{Prob}[Discrete Coded Cache Problem for AFOT Maximization]
\label{Prob:Codedcachingoptimization}
\begin{align}\label{}
  \max_{\mathbf m}\quad &\sum\limits_{j=1}^F p_jL[m_j] \nonumber\\
  s.t.\quad &\sum\limits_{j=1}^F \frac{m_j}n \leq M, \nonumber\\
 &m_j \in \mathcal{N}, \forall j\in \mathcal{F}. \label{eqn:discrete-constraint-2}
\end{align}
\end{Prob}

 Before solving Problem~\ref{Prob:Codedcachingoptimization}, we shall present a property of the optimal caching vector $\mathbf{m^*}$.

 \begin{Thm}\label{Thm:MorePopularMorePackets}
 The optimal caching vector $\mathbf{m^*}$ of Problem~\ref{Prob:Codedcachingoptimization} satisfies that $m_1^*\geq m_2^*\geq \cdots \geq m_F^*$, for any given file popularities $p_1 \geq p_2 \geq \cdots\geq p_F$.
 \end{Thm}
\begin{proof}Please see Appendix E.
\end{proof}
Theorem \ref{Thm:MorePopularMorePackets} matches the common intuition that more cache size should be allocated to more popular contents when the total cache size is limited.

Note that Problem~\ref{Prob:Codedcachingoptimization} is a multiple-choice knapsack problem (MCKP) \cite[Chapter 11]{kellerer2004}, which is known to be NP-hard.
The problem can be interpreted as follows.
Consider a knapsack with capacity of $Mn$ packets and a set of $F$ files each containing $n$ packets.
The size of each file packet is the same but the profit of each file packet varies according to the number of packets as well as the index of file.
Based on the properties of $L[m_j]$ analyzed in Section~\ref{subsec:AnalysisandPropertiesofFOT} and the property of $\mathbf{m^*}$ in Theorem~\ref{Thm:MorePopularMorePackets}, we design a greedy but optimal algorithm as outlined in Algorithm~\ref{alg:MKP} to solve Problem~\ref{Prob:Codedcachingoptimization}.
\begin{algorithm}[tb]
\caption{Optimal Algorithm for Problem~\ref{Prob:Codedcachingoptimization}}
\label{alg:MKP}
\begin{algorithmic}[1]
\STATE \textbf{initialize}
  \begin{itemize}
  \item Store $n$ packets for each of the $M$ most popular files in each SBS, i.e., $\mathbf m=[\underbrace {n,\cdots,n}_M, \underbrace {0,\cdots,0}_{F-M}]$;
      \item Construct a offloaded traffic difference table for the given $n$;
      \item Set $i=M+1$;
  \end{itemize}
\WHILE{$i\leq F$}
\label{algprithm1:2}
\IF{$m_i<n$}
\label{algprithm1:3}
\STATE Find the forward FOT difference $\delta_i$ in the table when $m_i$ increases by one;
\FORALL{$j \in \{1,2,\ldots,i-1\}$}
\STATE Find the backward FOT difference $\delta_j$ in the table when $m_j$ decreases by one;
\ENDFOR
\STATE Let $j'=\mathop{\argmin}_{j \in \{1,2,\ldots,i-1\}}{p_j\delta_j}$;
\IF{$p_i\delta_i>p_{j'}\delta_{j'}$}
\STATE Set $m_i = m_i+1$, $m_{j'} = m_{j'}-1$ and goto step \ref{algprithm1:3};
\ELSE
\STATE When $m_i>0$, goto step \ref{algprithm1:16};
\STATE When $m_i=0$, terminate the algorithm; \label{algprithm1:17}
\ENDIF
\ENDIF
\STATE Set $i=i+1$; \label{algprithm1:16}
\ENDWHILE
\end{algorithmic}
\end{algorithm}

Algorithm~\ref{alg:MKP} initializes the caching vector $\mathbf m$ by storing all the $n$ packets of the $M$ most popular files.
It then gradually updates the $Mn$ cached packets in a greedy manner by replacing an exiting packet of a more popular file in the knapsack with a new packet of a less popular file towards the direction of maximizing the total profit, i.e., AFOT.
More specifically, if the increased profit, i.e., $p_i \left(L[m_i+1] - L[m_i]\right)$ brought by adding one more packet from the $i$-th popular content into the cache is larger than the minimum of decreased profits, i.e., $\mathop{\min}_{j \in \{1,2,\ldots,i-1\}}{p_j \left(L[m_j] - L[m_j-1]\right) }$ caused by discarding one packet from any of the previous $i-1$ stored contents, the algorithm uses the new packet of $W_i$ to replace the existing packet with the minimum decreased profit, updates the caching vector, and continues to search new possible replacement.
The total profit monotonically increases in each update due to the monotonicity of $L[m_j]$ by Corollary~\ref{Cor:Delta-mj-decreasing}.
Also, each discarded packet during the replacement process will never be added again (hence never optimal) since its weighted FOT difference is always smaller than that of any newly added packet during subsequent processes as well as any remaining packet in the knapsack due to the concavity of $L[m_j]$ by Corollary~\ref{Cor:Delta-mj-decreasing}.
Thus, when the algorithm terminates, all the possible replacements are searched and the globally optimal caching vector $\mathbf m^*$ is obtained.

The offloaded traffic difference table constructed at the beginning of algorithm alleviates the computation burden of the forward/backward FOT differences during the iterations.
The maximum number of caching vector updates in Algorithm~\ref{alg:MKP} is $\min\left\{(n-1)(F-M), \sum_{i=1}^{n-1}\frac{nM}{i+1}\right\}$.

Based on Algorithm~\ref{alg:MKP}, we show that under certain conditions, the optimal coded caching will degenerate to the conventional most popular caching (MPC).
\begin{Thm}\label{Thm:Degenrate-MPC}
 For $n\geq 2$, when the following condition holds:
 \begin{align}\label{eqn:condition-degenrate-MPC}
   \frac{p_M}{p_{M+1}} \geq \frac{\sum_{k=1}^nC_k(\tau)}{C_1(\tau)-C_2(\tau)},
 \end{align}
 the optimal coded caching vector $\mathbf m^*$ of Problem~\ref{Prob:Codedcachingoptimization} is $\mathbf m^*=[\underbrace {n,\cdots,n}_M,\underbrace {0,\cdots,0}_{F-M}]$.
 \end{Thm}
\begin{proof}
By observing Algorithm~\ref{alg:MKP}, it is seen that the optimal $\mathbf m^*$ remains as the initialized value in step 1 when the algorithm terminates at step 13 with $i=M+1$ and $m_{M+1}=0$. Thus, to prove this theorem, it is equivalent to proving the condition $p_{M} \delta_{M} \geq p_{M+1} \delta_{M+1}$ by step 9, i,e., $p_M\left(L[n]-L[n-1]\right)\geq p_{M+1}L[1]$.
Hence, we have
\begin{align}\label{eqn:proof-condition-degenrate-MPC}
  &p_M\left(L[n]-L[n-1]\right) - p_{M+1}L[1] \nonumber\\
  =&p_M\left(C_1(\tau) - \left(\frac{n-1}n\sum\limits_{k=1}^{2}C_k(\tau) +\left(1-2\frac{n-1}{n}\right)C_2(\tau)\right)\right) - p_{M+1}\left(\frac1n \sum\limits_{k=1}^nC_k(\tau)\right) \nonumber\\
  =&\frac1n \left(p_M \left(C_1(\tau)-C_2(\tau)\right) - p_{M+1}\left(\sum\limits_{k=1}^nC_k(\tau)\right)\right) \geq 0.
\end{align}
Transforming the condition~\eqref{eqn:proof-condition-degenrate-MPC}, we can get the condition~\eqref{eqn:condition-degenrate-MPC}.
\end{proof}

To gain more insights from Theorem~\ref{Thm:Degenrate-MPC}, we take a closer look at the condition~\eqref{eqn:condition-degenrate-MPC} and recall the definition of $C_k(\tau)$ in~\eqref{eqn:c_k_tau}.
We then have the following corollary.
\begin{Cor}
\label{Cor:tau-gamma-MPC}
   For $n\geq 2$ and $\alpha > 2$, we obtain that
   \begin{enumerate}
     \item if $p_M > p_{M+1}$, there exists a decoding threshold $\tau_0>0$ such that the coded caching degenerates into MPC in terms of the AFOT performance, at all $\tau \geq \tau_0$;
     \item if the file popularity follows the Zipf distribution with shape parameter $\gamma$, there exists a shape parameter $\gamma_0>0$ such that the optimal coded caching degenerates into MPC in terms of the AFOT performance, at all $\gamma \geq \gamma_0$, when $\tau>0$.
   \end{enumerate}
\end{Cor}
\begin{proof}Please see Appendix F.
\end{proof}


\subsection{Continuous Problem for AFOT Optimization} \label{subsec:ContinuousOptimizationonMaximizingAFOT}
In this subsection, we study the maximization of AFOT by letting $n \to \infty$.
Based on that, we then design a low-complexity sub-optimal algorithm to solve Problem~\ref{Prob:Codedcachingoptimization}.
The continuous cache optimization problem for $n\to\infty$ is formulated as follows.
\begin{Prob}[Continuous Coded Cache Problem for AFOT Maximization]
\label{Prob:ContinuousOptimizationProblem}
\begin{align}\label{}
  \max_{\mathbf x}\quad &\sum\limits_{j=1}^F p_j L(x_j) \nonumber\\
  s.t.\quad &\sum\limits_{j=1}^F x_j \leq M, \label{eqn:constraint-continuous-1}\nonumber\\
  & 0 \leq x_j \leq 1, \forall j\in \mathcal{F}.
\end{align}
\end{Prob}
Since $L(x_j)$ is a concave function by Lemma~\ref{Lem:R-xj-increase-concave},  Problem~\ref{Prob:ContinuousOptimizationProblem} is a convex problem and hence can be solved efficiently by various methods, e.g., the interior point method.
The optimal AFOT of Problem~\ref{Prob:ContinuousOptimizationProblem} provides an upper bound of that of Problem~\ref{Prob:Codedcachingoptimization} for any $n$.

Next, we propose a low-complexity algorithm for Problem~\ref{Prob:Codedcachingoptimization} based on the optimal solution of $\mathbf x^*$ in Problem~\ref{Prob:ContinuousOptimizationProblem}, denoted as $\mathbf{x^*} \triangleq [x_1^*,x_2^*,\cdots,x_F^*]$.
At the initial step, let $m_j=\lceil nx_j^*\rceil$, $\forall j \in \mathcal F$.
Since the problem is convex, we have $\sum\limits_{j=1}^F x_j^* = M$.
Thus, we have $\sum\limits_{j=1}^F\frac{m_j}n \geq M$. Define the number of exceeding packets as
\begin{align}\label{eqn:N_e}
 N_e\triangleq \sum\limits_{j=1}^F m_j-Mn.
\end{align}
Now we refine $m_j$ by removing $N_e$ packets from each SBS to satisfy the cache size constraint. This process is outlined in Algorithm~\ref{alg:Near-optimal-Solution}.
\begin{algorithm}[tb]
\caption{Low-complexity Algorithm for Problem~\ref{Prob:Codedcachingoptimization}}
\label{alg:Near-optimal-Solution}
\begin{algorithmic}[1]
\STATE \textbf{initialize}
  \begin{itemize}
      \item Find the optimal $\mathbf{x}^*$ by solving Problem~\ref{Prob:ContinuousOptimizationProblem};
      \item Set $\mathbf{m}=\left[\lceil nx_1^*\rceil,\lceil nx_2^*\rceil,\cdots, \lceil nx_F^*\rceil\right]$;
      \item Compute $N_e$ as \eqref{eqn:N_e};
      \item Construct a offloaded traffic difference table for given $n$;
      \item Set $i=1$;
  \end{itemize}
\WHILE{$i\leq N_e$}
 \FORALL {$j\in \mathcal{F}$ and $m_j \ne 0$}
  \STATE Find the backward FOT difference $\delta_j$ in the table when $m_j$ decreases by one;
 \ENDFOR
 \STATE Let $j'=\mathop{\argmin}_{j \in \{i|m_i \ne 0\}}{p_j\delta_j}$; \\
 \STATE Set $m_{j'} = m_{j'}-1$;
 \STATE Set $i=i+1$;
\ENDWHILE
\end{algorithmic}
\end{algorithm}

Algorithm~\ref{alg:Near-optimal-Solution} initializes the caching vector $\mathbf m$ to be $m_j = \lceil nx_j^*\rceil$, for all $j\in \mathcal F$, which exceeds the cache size constraint by $N_e$ in
\eqref{eqn:N_e}.
Then, it gradually decreases coded packets one by one until the cache size constraint is satisfied.
Each time only the coded packet of the file with the minimum profit, i.e., $\mathop{\min}_{j \in \mathcal{F}}{p_j\delta_j}$, will be discarded.
The optimality of the algorithm cannot be guaranteed since if there exists a file $W_j$ whose initial $m_j$ is smaller than the global optimal $m_j^*$, then $m_j^*$ can never be reached since we only decrease packets for files.
However, if each initial $m_j\geq m_j^*$, the solution of Algorithm 2 is optimal.
Due to the monotonicity and concavity of $L[m_j]$ in Corollary~\ref{Cor:Delta-mj-decreasing}, the discarded packet causes the minimum loss in weighted FOT at each caching vector update among all the packets in each SBS.
Continuing Algorithm~\ref{alg:Near-optimal-Solution}, the packets with minimum loss in weighted FOT are discarded one by one.
Thus, Algorithm~\ref{alg:Near-optimal-Solution} can obtain the optimal solution with the maximum AFOT when the discarding process terminates.

Note that the maximum possible number of caching vector updates in Algorithm~\ref{alg:Near-optimal-Solution} is $F$, which is independent of $n$.
Hence Algorithm~\ref{alg:Near-optimal-Solution} is more computationally efficient than Algorithm~\ref{alg:MKP}.
We shall provide more detailed comparisons with Algorithm~\ref{alg:MKP} in Section~\ref{subsec:Complexity-Comparison}.

\subsection{Discrete Problem for AER Optimization}
\label{subsec:OptimizationonMaximizingAER}
In this subsection, we aim to maximize AER in~\eqref{eqn:definition-average-ergodic-rate} by optimizing the caching vector $\mathbf m$ subject to the cache size constraint~\eqref{eqn:discrete-constraint-1} at a given finite $n$.
This is formulated as:

\begin{Prob}[Coded Cache Problem for AER Maximization]
\label{Prob:CodedcachingonAverageErgodicRate}
\begin{align}\label{}
  \max_{\mathbf m}\quad &\sum\limits_{j=1}^F p_jR[m_j] \nonumber\\
  s.t.\quad &\sum\limits_{j=1}^F \frac{m_j}n \leq M, \nonumber\\
 &m_j \in \mathcal{N}, \forall j\in \mathcal{F}.
\end{align}
\end{Prob}

Monotonicity and concavity of $R[m_j]$ is unknown due to the integral expression of $R[m_j]$  in~\eqref{eqn:Ergodic-Rate}.
Thus, we transform it into a standard MCKP first, and then propose a heuristic algorithm to solve it.
\begin{Prob}[Standard MCKP for AER Maximization]
\label{Prob:MCKPonAverageErgodicRate}
\begin{align}\label{}
  \max_{\mathbf X}\quad &\sum\limits_{j=1}^F \sum\limits_{k=0}^n R_{jk}x_{jk} \nonumber\\
  s.t.\quad &\sum\limits_{j=1}^F \sum\limits_{k=0}^n w_{jk}x_{jk} \leq M, \nonumber\\
  &\sum\limits_{k=0}^n x_{jk} = 1, \forall j\in\mathcal{F}, \nonumber\\
  &x_{jk}\in\{0,1\}, \forall j\in\mathcal{F}, \forall k\in \mathcal{N},
\end{align}
where $\mathbf X = \{x_{j,k},j\in\mathcal{F},k\in \mathcal{N}\}$, $R_{jk} = p_jR[k]$, and $w_{jk} = \frac{k}{n}$.
\end{Prob}

As mentioned in \cite{kellerer2004}, we can relax the binary constraint on $x_{jk}$ and obtain the linear MCKP, which is a standard linear programming and can be solved efficiently.
Based on the optimal solution of the linear MCKP, we then propose a heuristic algorithm to solve the original Problem~\ref{Prob:CodedcachingonAverageErgodicRate} as shown in Algorithm~\ref{alg:Heuristic-Near-optimal-Solution}.
Note that we terminate the algorithm at step 7 in advance since the AER decreases after increasing one packet.

\begin{algorithm}[tb]
\caption{The Heuristic Algorithm for Problem~\ref{Prob:CodedcachingonAverageErgodicRate}}
\label{alg:Heuristic-Near-optimal-Solution}
\begin{algorithmic}[1]
\STATE \textbf{initialize}
  \begin{itemize}
      \item Find the optimal $x_{j,k}^*, j\in {\mathcal F}, k\in{\mathcal{N}}$ by solving the standard MCKP;
      \item Set $m_j = \mathop{\argmax}_{k \in \mathcal{N}}{x_{j,k}^*}, \forall j\in \mathcal F$;
      \item Compute the total packet size as $N_i = \sum_{j=1}^F\frac{m_j}n$;
  \end{itemize}
\IF {$N_i > M$}
 \STATE Apply the same method as Algorithm~\ref{alg:Near-optimal-Solution} to discard $n(N_i-M)$ packets;
\ELSE
 \STATE Let $j'=\mathop{\argmax}_{j\in \{i|m_i \ne n\}}{p_j\left(R[m_j+1]-R[m_j]\right)}$;
 \IF{$R[m_{j'}+1]-R[m_{j'}] < 0$}
  \STATE Terminate the algorithm;
  \ELSE
  \STATE Add one packet of $W_{j'}$ in each SBS;
  \STATE Increase $N_i$ by $\frac 1 n$ and goto step 5 unless $N_i = M$;
  \ENDIF
\ENDIF
\end{algorithmic}
\end{algorithm}

\section{Numerical Results}\label{sec:NumericalResults}
In this section, we illustrate the performance of our proposed coded caching scheme in SCNs using numerical examples.
Throughout this section, there are $F = 100$ files in total and the channel pathloss exponent is $\alpha = 4$.
The file popularity is modeled as the Zipf distribution with parameter $\gamma$.


\subsection{Successful Transmission Probability}
\label{subsec:SuccessfulTransmissionProbabilityValidation}
\begin{figure}[tb]
\begin{center}
  \includegraphics[width=9.5cm]{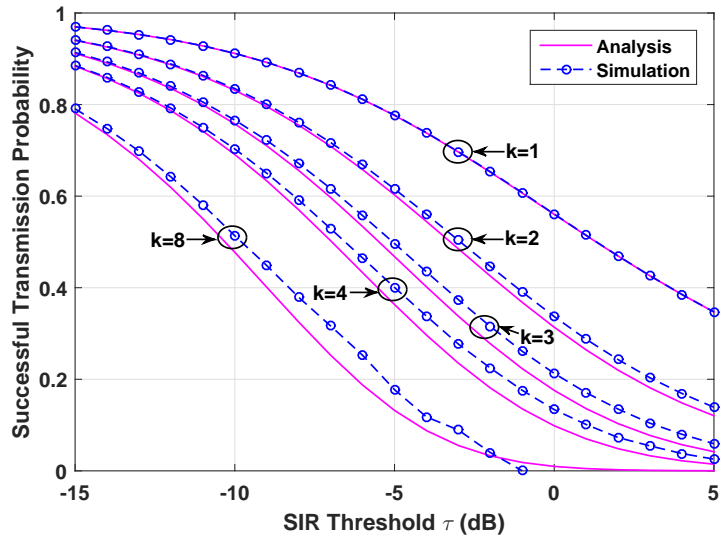}\\
  \caption{Successful transmission probability of the $k$-th nearest SBS with SIC receiver at $n=8$.}
  \label{Fig:The-simulations-of-coverage-probability}
\end{center}
\end{figure}
We first validate the analysis of the successful transmission probability $q_k(\tau)$ with the SIC receiver in Lemma~\ref{Lem:CodedTransmissionCoverageProbability} and its application in the analysis of FOT $L[m_j]$ in Lemma~\ref{Lem:BRofCodedCaching}.
The simulation is conducted over a square area of $4 \text{km} \times 4 \text{km}$ with SBS density $\lambda_b=10^{2}/\text{km}^2$.
Each simulation result is averaged over $10^6$ independent sets of SBS locations and channel realizations.
Note that the simulation results strictly follow the definition in \eqref{eqn:definition-q-k}.
That is, we only evaluate if the transmission of $\phi_k$ is successful when the transmissions from $\{\phi_1,\phi_2,\ldots,\phi_{k-1}\}$ are all successful.
The analytical result in Lemma~\ref{Lem:CodedTransmissionCoverageProbability} is however based on the assumption that the $\text{SIR}_k$'s are independent for different $k$ as mentioned in Appendix A.
Fig.~\ref{Fig:The-simulations-of-coverage-probability} shows that the analytical results and the simulation results of $q_k(\tau)$ match exactly when $k=1$. For $k\ge 2$ and low SIR threshold $\tau$, the results also match well.
For large $\tau$, the analytical results serve as a good lower bound of the actual simulation results in most cases.
It is also seen from Fig.~\ref{Fig:The-simulations-of-coverage-probability} that at large SIR thresholds, when the decoding layer $k$ increases, the successful transmission probability drops considerably.
This indicates that the overall system performance of SIC receiver will be dominated by the first decoding layer at large $\tau$.

\begin{figure}[tb]
\begin{center}
  \includegraphics[width=9.5cm]{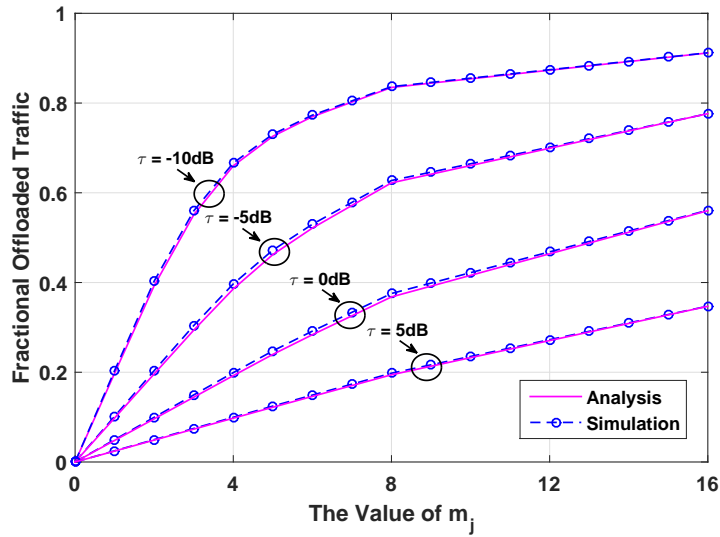}\\
  \caption{Fractional offloaded traffic $L[m_j]$ with $\alpha=4$, $n=16$.}
  \label{Fig:l-m-j-simulation-analysis}
\end{center}
\end{figure}
The comparison between the analytical $L[m_j]$ in~\eqref{eqn:R_c} and the simulated $L[m_j]$ by definition~\eqref{eqn:R-mj} is shown in Fig.~\ref{Fig:l-m-j-simulation-analysis}.
It is seen that the two sets of results match each other very well for all the considered $m_j$ and $\tau$.
This means that ignoring the dependency among the SIRs in different SIC decoding layers has little effects on the FOT performance, even though it can cause a visible gap on the successful transmission probability.
This observation is mainly because the overall FOT performance is dominated by the first two decoding layers, whose approximations are tighter than other layers as shown in Fig.~\ref{Fig:The-simulations-of-coverage-probability}.

\subsection{Algorithm~\ref{alg:MKP} vs. Algorithm~\ref{alg:Near-optimal-Solution} for AFOT maximization}
\label{subsec:Complexity-Comparison}
\begin{figure}[tb]
\begin{minipage}{0.486\linewidth}
\centering
\includegraphics[width=8.2cm]{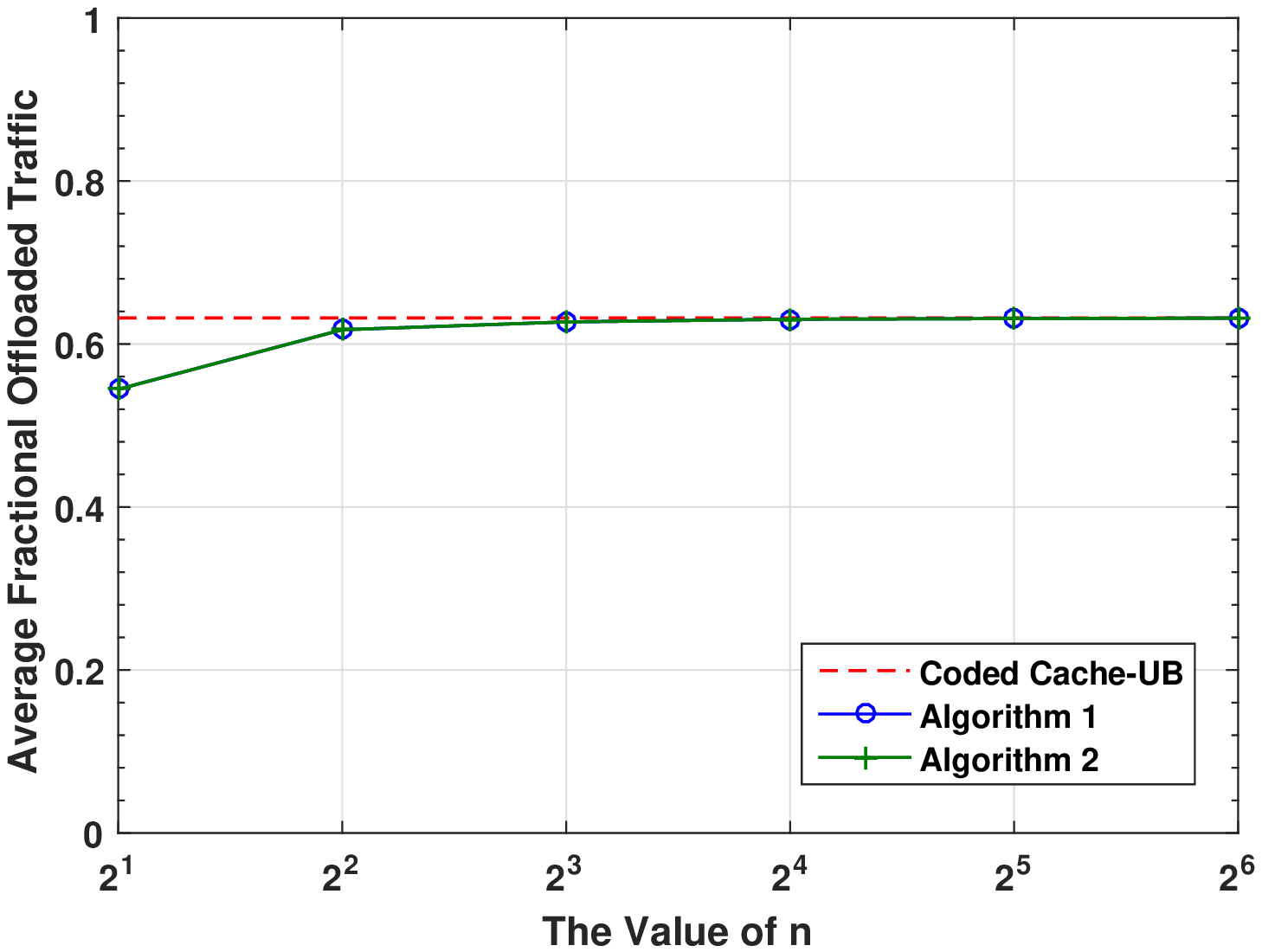}\\
  \caption{AFOT with respect to $n$, with $\gamma=0.6, \tau=-10\ \text{dB}, M=20$.}\label{Fig:PerformanceofIteration}
\end{minipage}
\hfill
\begin{minipage}{0.486\linewidth}
\centering
\includegraphics[width=8.2cm]{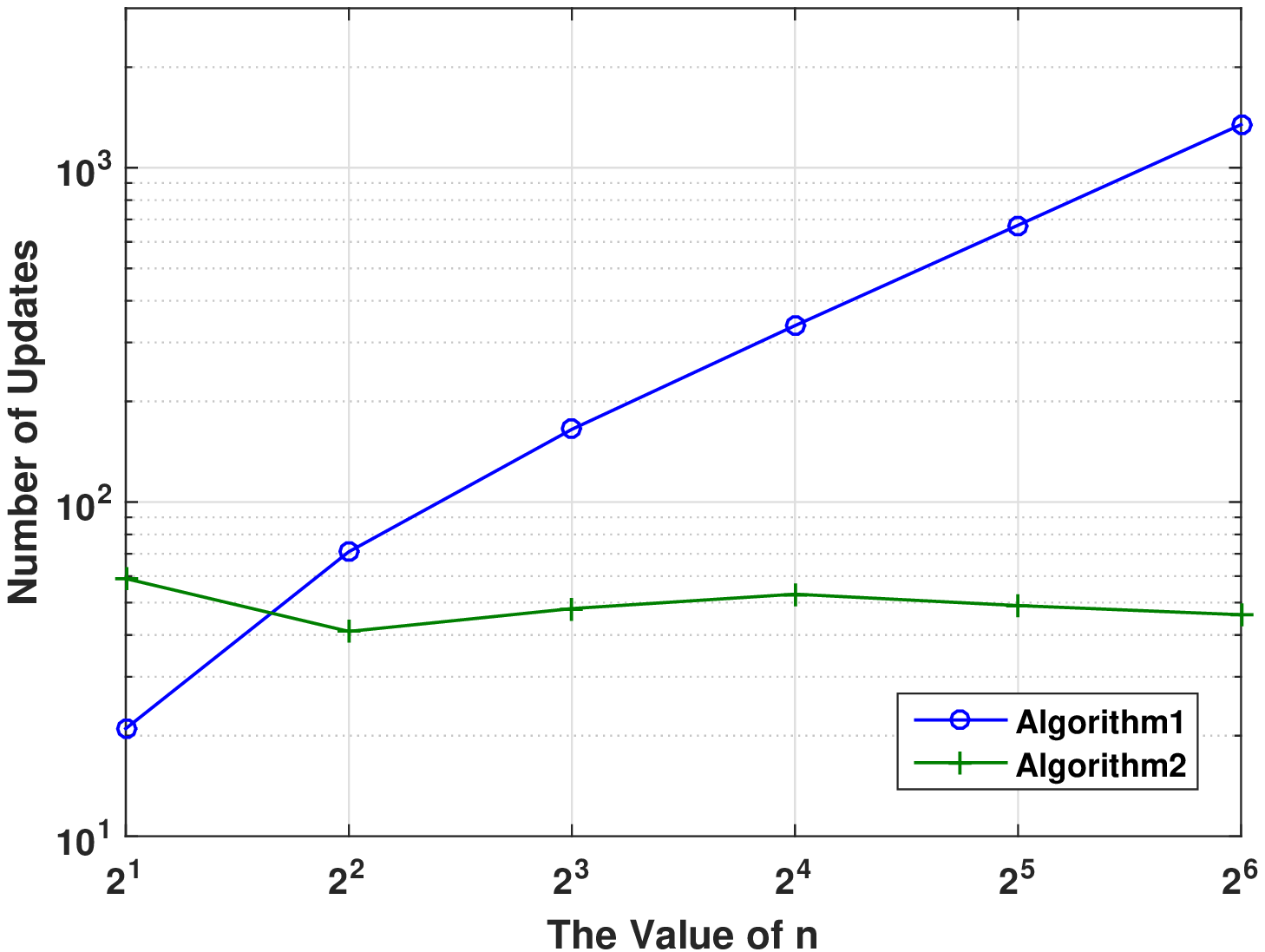}\\
  \caption{Number of updates for $\mathbf{m}$ when $n$ changes, with $\gamma=0.6, \tau=-10\ \text{dB}, M=20$.}\label{Fig:NumberofIteration}
\end{minipage}

\begin{minipage}{0.486\linewidth}
\centering
\includegraphics[width=8.2cm]{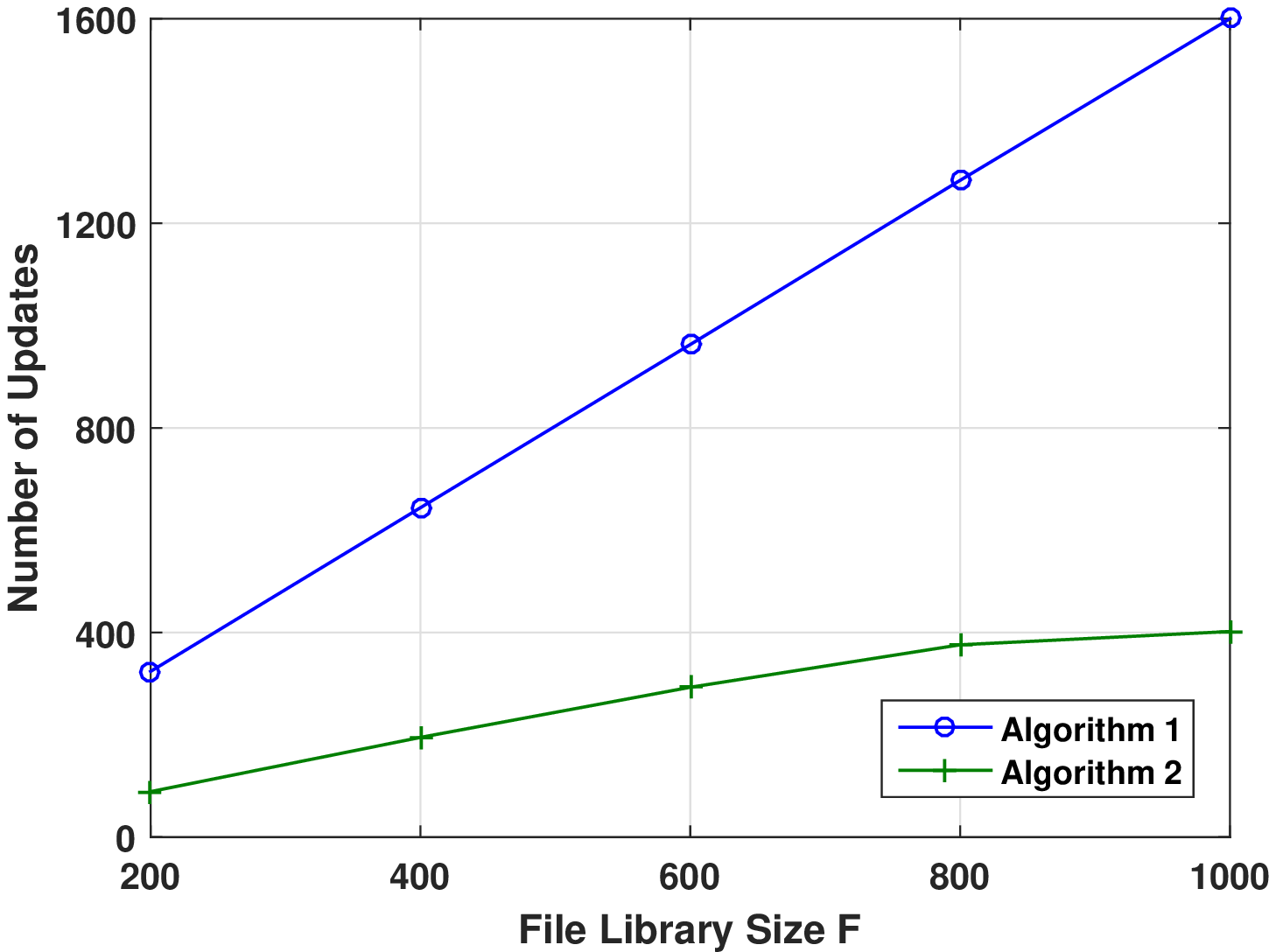}\\
  \caption{Number of updates for $\mathbf{m}$ when $F$ changes, with $\gamma=0.6, \tau=-10\ \text{dB}, M=0.2*F, n=8$.}\label{Fig:NumberofIteration-F}
\end{minipage}
\hfill
\begin{minipage}{0.486\linewidth}
\centering
\includegraphics[width=8.2cm]{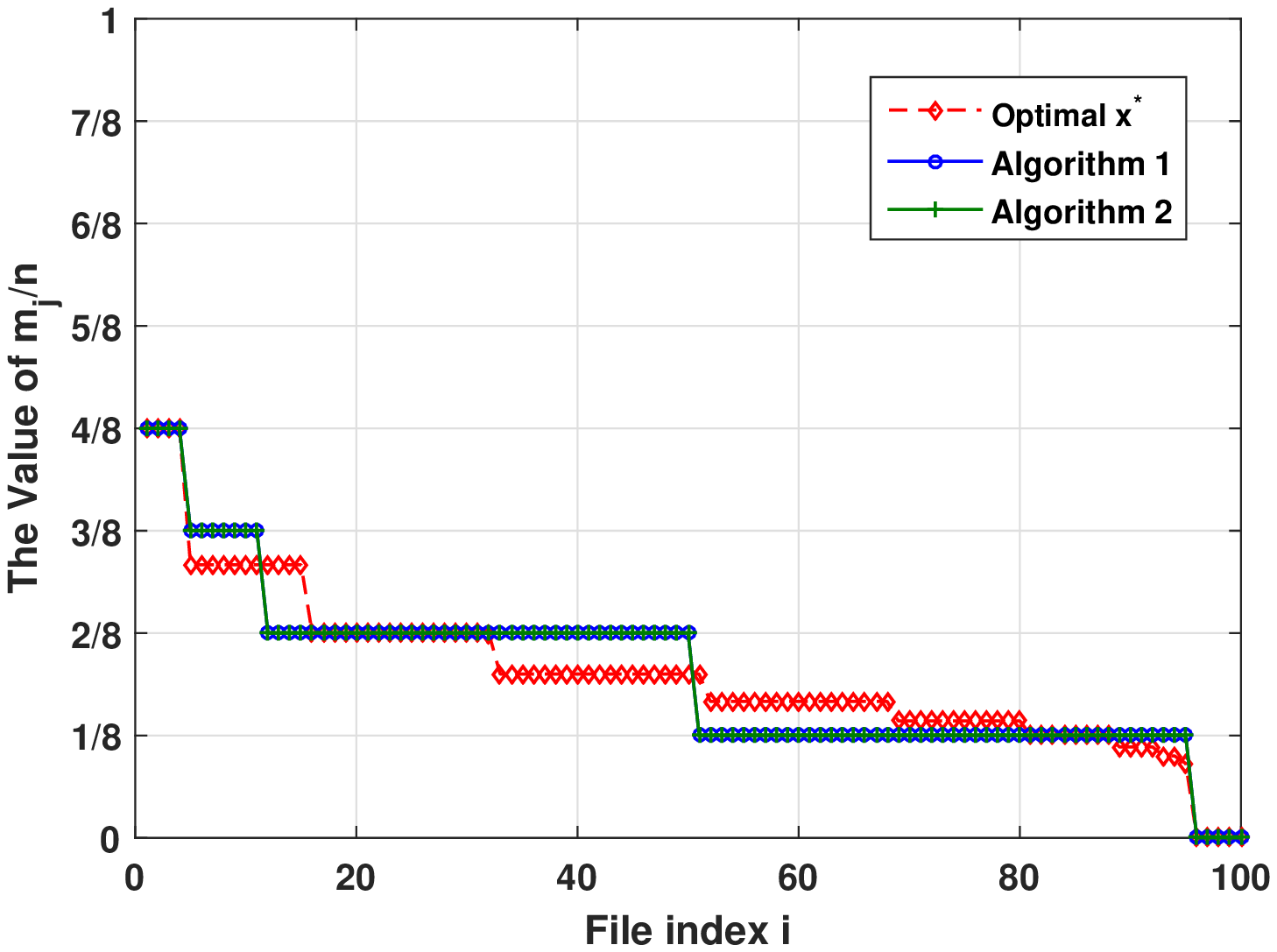}\\
  \caption{The caching vector $\mathbf m$, with $\gamma=0.6, \tau=-10\ \text{dB}, M=20, n=8$.}\label{Fig:m-array}
\end{minipage}
\end{figure}
In this subsection, we provide the comparison between Algorithm~\ref{alg:MKP} and Algorithm~\ref{alg:Near-optimal-Solution} for the AFOT maximization  problem, i.e., Problem~\ref{Prob:Codedcachingoptimization}.
Fig.~\ref{Fig:PerformanceofIteration} shows the AFOT performance at different coding parameter $n$.
As a performance upper bound, the optimal solution of the continuous problem, Problem~\ref{Prob:ContinuousOptimizationProblem} is also shown and denoted as ``Coded Cache-UB''.
It is seen that Algorithm~\ref{alg:Near-optimal-Solution} performs almost identical to Algorithm~\ref{alg:MKP} for all the considered $n$.
When $n$ increases, they both approach to Coded Cache-UB.
Fig.~\ref{Fig:NumberofIteration} and Fig.~\ref{Fig:NumberofIteration-F} compare the computational complexity in terms of the number of caching vector updates at different coding parameter $n$ and different file number $F$, respectively.
The complexity of solving Problem~\ref{Prob:ContinuousOptimizationProblem} in Algorithm~\ref{alg:Near-optimal-Solution} is ignored for simplicity.
Fig.~\ref{Fig:NumberofIteration} shows that as $n$ increases, the complexity of Algorithm~\ref{alg:MKP} increases (still manageable), but the complexity of Algorithm~\ref{alg:Near-optimal-Solution} almost remains unchanged.
Fig.~\ref{Fig:NumberofIteration-F} shows that as $F$ increases, the complexity of Algorithm~\ref{alg:Near-optimal-Solution} increases much slower than that of Algorithm~\ref{alg:MKP}.

Fig.~\ref{Fig:m-array} compares the solution values of the obtained caching vector $\mathbf m$ in different algorithms when $n=8$, where the optimal solution $\mathbf x^*$ of the continuous problem, Problem~\ref{Prob:ContinuousOptimizationProblem}, is also included for comparison.
It is seen that the solution of Algorithm~\ref{alg:Near-optimal-Solution} is identical to that of Algorithm~\ref{alg:MKP}.
Fig.~\ref{Fig:m-array} also shows that more coded packets should be cached for more popular files, which justifies Theorem~\ref{Thm:MorePopularMorePackets}.

\begin{Rem}
In general, the encoding and decoding complexity of MDS codes or RLNC increases as the total number of fragments that each file is split into, i.e., $n$, increases.
However, the actual number of fragments that a file is split into does not have to be $n$, depending on the optimized caching parameter $m_j^*$.
For example, we have $\frac{m_1^*}{n}=\frac{4}{8}=\frac{1}{2}$ for file $W_1$ from Fig.~\ref{Fig:m-array}. Then, $W_1$ only needs to be split into $2$ fragments for encoding and each SBS stores one coded packet.
Also, we have $\frac{m_{20}^*}{n}=\frac{2}{8}=\frac{1}{4}$ for file $W_{20}$. Then, $W_{20}$ only needs to be split into $4$ fragments for encoding and each SBS stores one coded packet.
\end{Rem}

From the above results, we conclude that Algorithm~\ref{alg:Near-optimal-Solution} achieves almost the same performance as Algorithm~\ref{alg:MKP} but with much lower complexity, especially for large $n$ and $F$.
Nevertheless, we still use Algorithm~\ref{alg:MKP} for the AFOT maximization in the rest of this section since it obtains the globally optimal solution with manageable complexity.



\subsection{Effects of coding parameter $n$}
\label{subsec:Effectof-n}
Fig.~\ref{Fig:averageMBStrafficload-n-M} depicts the AFOT performance with respect to different $n$ when the cache size $M$ varies. In the special case with $n=1$, the coded caching reduces to uncoded caching which only stores the $M$ most popular files in each SBS, namely $m_j = n$, for $j\leq M$, and $m_j = 0$, for $j>M$.
It is seen from Fig.~\ref{Fig:averageMBStrafficload-n-M} that the AFOT increases dramatically when $n$ increases from $1$ to $4$ but the gain diminishes quickly when $n$ reaches $4$.
When $n=8$, the performance is very close to the upper bound determined the optimal solution of Problem~\ref{Prob:ContinuousOptimizationProblem}.
Note that the coding parameter $n$ is the maximum possible number of decoding layers for SIC receiver and hence determines the maximum receiver complexity.
It also affects the computational complexity of Algorithm~\ref{alg:MKP} as shown in Fig.~\ref{Fig:NumberofIteration}.
Thus, $n$ cannot be too large in practical systems.
The results in Fig.~\ref{Fig:averageMBStrafficload-n-M} suggest that it is good enough to let $n=8$.
In the rest of this section, we fix $n=8$ when considering the AFOT performance.
\begin{figure}
\begin{center}
  \includegraphics[width=9.5cm]{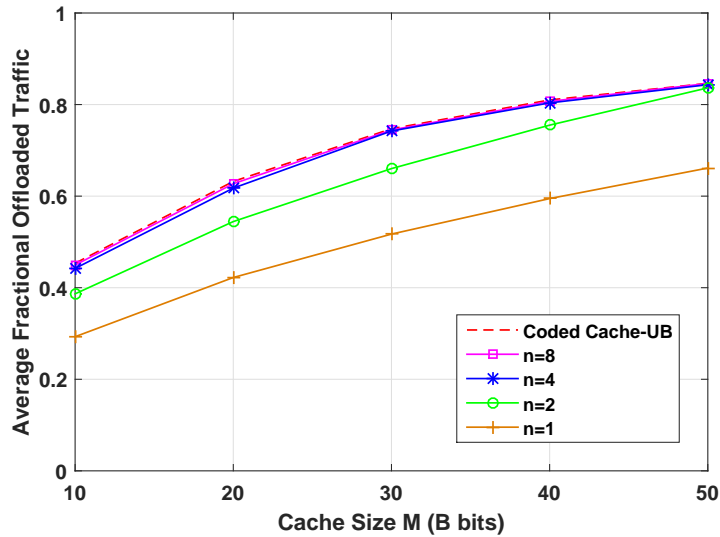}\\
  \caption{AFOT with respect to $n$ and $M$, with $\gamma=0.6, \tau=-10\ \text{dB}$.}
  \label{Fig:averageMBStrafficload-n-M}
\end{center}
\end{figure}


\subsection{Coded Caching vs. Uncoded Caching on AFOT}
\label{subsec:SuperiorPerformanceofCodedCaching}

We show the performance of coded caching in comparison with uncoded caching on the AFOT under different parameter settings.
Two uncoded caching schemes are considered as benchmarks, namely most popular caching (MPC) and optimal probabilistic caching (OPC).

\subsubsection{Most Popular Caching (MPC)}
\label{subsubsec:MPC}
 This is the special case of coded caching when $n=1$. Specifically, we store the $M$ most popular files in each SBS.
 If $u_0$ requests any of these cached files, it only connects with the nearest SBS. If the transmission of the nearest SBS is successful, the user request is satisfied.
 If not, $u_0$ retrieves the file from an MBS.
 If $u_0$ requests any of the other $F-M$ uncached files, the file will be transmitted directly by the MBS. Therefore, the AFOT, denoted as $L_M$, can be calculated as $L_M =q_1\sum_{j=1}^M p_j$.

\subsubsection{Optimal Probabilistic Caching (OPC)}
\label{subsubsec:OPC}
In the probabilistic caching strategy, each SBS caches $W_j$ independently with probability $b_j$. The vector $\mathbf{b} \triangleq [b_1,b_2,\cdots,b_F]$ is referred to as the \emph{caching probability vector} with constraints as $\sum\limits_{j=1}^F b_j \leq M$.
When a user submits a request for $W_j$, it will be associated with the strongest SBS (in terms of the average received signal strength) that has cached the file. The probability vector $\mathbf{b}$ can be optimized to maximize the AFOT.
Note that this problem has been studied in \cite{chen}.
We refer readers to \cite[Theorem 3]{chen}.

\begin{figure}
\begin{center}
\includegraphics[width=9.5cm]{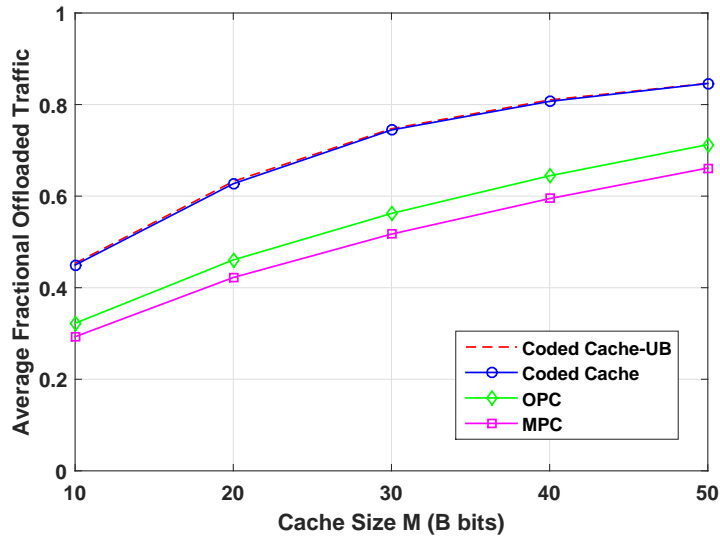}\\
\caption{AFOT with respect to cache size $M$, with $\gamma=0.6, \tau=-10\ \text{dB}.$}\label{Fig:M_0.6_0.1_20}
\end{center}
\end{figure}
Fig.~\ref{Fig:M_0.6_0.1_20} illustrates the AFOT performance with respect to cache size $M$.
It is observed that the proposed coded caching scheme performs significantly better than the two benchmarks.  This is because coded caching can fully exploit the accumulated cache size in the network. MPC, on the other hand, only enjoys the local caching gain since the same contents are cached in each SBS. OPC improves upon MPC by exploiting the randomness in the probabilistic caching but the gain is limited.

\begin{figure}
\begin{center}
\includegraphics[width=9.5cm]{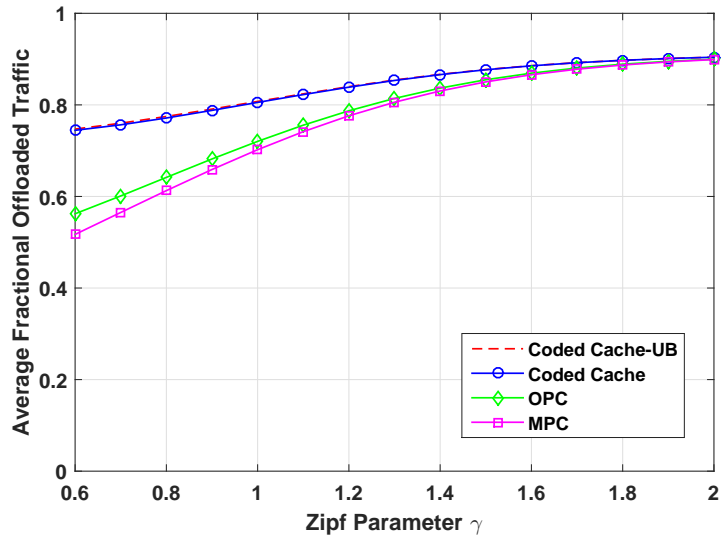}\\
\caption{AFOT with respect to Zipf parameter $\gamma$, with $\tau=-10\ \text{dB}, M=30.$}\label{Fig:s_30_0.1_20}
\end{center}
\end{figure}
Fig.~\ref{Fig:s_30_0.1_20} illustrates the AFOT performance with respect to Zipf parameter $\gamma$.
It is again observed that coded caching has much better performance than MPC and OPC, especially when $\gamma$ is small $(\gamma < 1.2)$.
In the high $\gamma$ region $(\gamma > 1.8)$, most of the user requests are for the few popular files in $\mathcal{W}$.
Both coded caching and OPC degenerate to MPC, which matches with the finding in Corollary~\ref{Cor:tau-gamma-MPC}.

\begin{figure}
\begin{center}
\includegraphics[width=9.5cm]{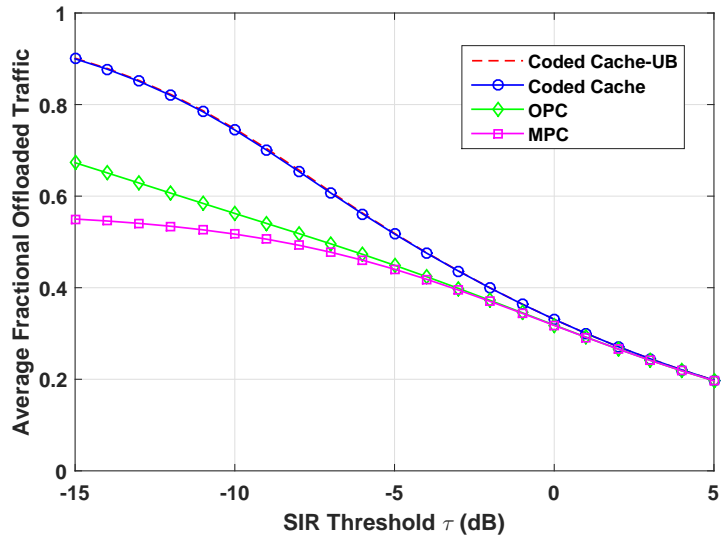}\\
\caption{AFOT with respect to SIR threshold $\tau$, with $\gamma=0.6, M=30.$}\label{Fig:tau_30_0.6_20}
\end{center}
\end{figure}
In Fig.~\ref{Fig:tau_30_0.6_20}, we illustrate the AFOT performance with respect to the target SIR threshold $\tau$.
 In the small $\tau$ region, coded caching has much better performance due to that the accumulated cache size of the nearest SBSs are used to store different contents.
 When $\tau$ increases, the performances of coded and uncoded caching schemes all converge as expected from Corollary~\ref{Cor:tau-gamma-MPC}.
 This is because when $\tau$ is large, the successful transmission probability for the nearest SBS in coded caching is low as shown in Fig.~\ref{Fig:The-simulations-of-coverage-probability}.
 As a result, the chance for the SIC-based receiver to decode the second or higher order signal is much lower.
 Thus, the accumulated cache size cannot be fully exploited by coded caching.
 In fact, for both coded caching and OPC, most of the cache size in each SBS is used to store the most popular contents when the successful transmission probability for the nearest SBS has been low.

\subsection{Coded Caching vs. MPC on AER}
\label{subsec:SuperiorPerformanceofCodedCachingAverageErgodicRate}
\begin{figure}[tb]
\begin{center}
\includegraphics[width=9.5cm]{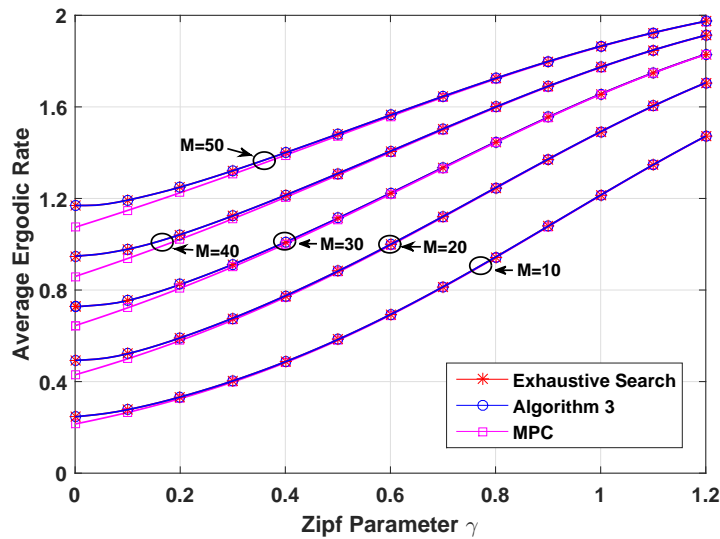}\\
\caption{The performance of average ergodic rate with respect to $\gamma$ and $M$, with $n=4$.}\label{Fig:AverageErgodicRate}
\end{center}
\end{figure}

Finally, in Fig.~\ref{Fig:AverageErgodicRate}, we illustrate the AER performance of coded caching in comparison with the conventional MPC at different Zipf parameter $\gamma$ and different cache size $M$.
The optimal solution of Problem~\ref{Prob:CodedcachingonAverageErgodicRate}, the coded cache problem for AER maximization,  obtained via exhaustive search is also presented for comparison.
It is first seen that Algorithm~\ref{alg:Heuristic-Near-optimal-Solution}, the proposed heuristic algorithm to solve Problem~\ref{Prob:CodedcachingonAverageErgodicRate}, performs almost identical to the optimal solution.
It is then seen that coded caching outperforms MPC when $\gamma \to 0$, i.e., the file popularity distribution is close to uniform, but the gain vanishes when $\gamma>0.2$.
This is because the ergodic rate of each file (if cached) in coded caching is determined by the worst SIR among all decoding layers as defined in~\eqref{eqn:definition-ergodic-rate} while the ergodic rate of each file (if cached) in MPC is always determined the nearest SBS.
As such, when the Zipf parameter $\gamma$ is large, the gain from the high cache hit as a result of exploiting the accumulated cache size in coded caching is not enough to compensate the rate loss due to SIC receiver. Thus, coded caching will degenerate to MPC when the file popularity distribution is non-uniform.

 \section{Conclusions}\label{sec:Conclusion}
In this paper, we investigated the optimal coded caching in cache-enabled SCNs.
We proposed a content delivery framework for coded caching where multiple SBSs transmit concurrently upon a user request and the user decodes independent coded packets from these SBSs with SIC-based receiver.
We obtained a closed-form expression for the success probability of each decoding layer with SIC receiver in the high SNR region and a closed-form expression of the average fractional offloaded traffic.
We also obtained a tractable expression of the average ergodic rate.
Then, a greedy-based optimal algorithm was proposed to solve the cache placement problem for AFOT maximization.
We also obtained an upper bound assuming a limiting case that the files are split into infinite parts.
Based on that, a low-complexity algorithm with high performance was proposed for the cache placement problem.
To maximize AER, we formulated the coded cache problem as a standard MCKP and proposed a heuristic algorithm to solve it.
Numerical results showed that coded caching with SIC receiver significantly outperforms uncoded caching in terms of AFOT performance at low SIR thresholds, even when number of segments for file splitting is small.
It was shown that coded caching degenerates to MPC in terms of AER performance unless the file popularity distribution tends to be uniform.

\section*{Appendix A: Proof of Lemma~\ref{Lem:CodedTransmissionCoverageProbability}}
In general, $\text{SIR}_k$'s in \eqref{eqn:SIR-C} are correlated for different $\phi_k$ in the SIC receiver.
To make the analysis more tractable, we assume that the events $\{\text{SIR}_k\geq \tau\}$ are independent for all $k\in\mathcal{N}^+$ as in \cite{Quek2014}.
As shown in Section~\ref{subsec:SuccessfulTransmissionProbabilityValidation}, ignoring such correlations has little effects on the system performance as in \cite{Quek2014}.
Hence, the successful transmission probability for the $k$-th nearest SBS is
\begin{align}\label{eqn:q-k-proof}
 q_k(\tau) &= \text{Pr}\left[\text{SIR}_k\geq \tau \bigg| \bigcap_{i=1,2,\ldots,k-1}\text{SIR}_i\geq \tau\right]
 \overset{(a)}{\approx} \text{Pr}\left[\text{SIR}_k\geq \tau\right] \nonumber\\
 &\overset{(b)}{=} {\left(1+\tau^{\frac2{\alpha}} \int_{\tau^{-\frac2{\alpha}}}^{\infty}\frac1{1+w^{\frac{\alpha}2}}\mathrm{d}w\right)}^{-k}
 \overset{(c)}{=}\left(1+\tau^{\frac2{\alpha}}\int_{\tau^{-\frac1{\alpha}}}^{^{\infty}}\frac{2t}{1+t^{\alpha}}\mathrm{d}t\right)^{-k} \nonumber\\
 &\overset{(d)}{=}\left(1+\frac2{\alpha}\tau^{\frac2{\alpha}}\int_{\frac1{1+\tau}}^1 u^{\frac2{\alpha}-1}(1-u)^{-\frac2{\alpha}}\mathrm{d}u\right)^{-k}\nonumber\\
 &=\left(1+\frac2{\alpha}\tau^{\frac2{\alpha}}B'\left(\frac2{\alpha},1-\frac2{\alpha},\frac1{1+\tau}\right)\right)^{-k}.
\end{align}
Here, $(a)$ is due to the independence assumption as in \cite{Quek2014};
$(b)$ follows from \cite[Lemma 3]{Quek2014} by letting $\tau=\eta_t$;
$(c)$ is due to the change of variable $t=w^{\frac12}$; and $(d)$ is due to the change of variable $u=\frac1{1+t^{-\alpha}}$.
Note that the changes of variables in $(c)$ and $(d)$ are to convert the unbounded integral in $(b)$ to the definite integral, which is numerically more convenient.

\section*{Appendix B: Proof of Lemma~\ref{Lem:BRofCodedCaching}}
 When $m_j=0$, there is no packet of $W_j$ in each SBS. Therefore, FOT is equal to $0$.

 When $m_j \in \mathcal{N^+}$, we have FOT as follows.
 \begin{align}\label{eqn:R_c-sum-1}
    L[m_j]
    =& \sum\limits_{k=1}^{\lceil\frac{n}{m_j}\rceil} \frac{(k-1)m_j}{n} \left(1-\text{Pr} \left[\text{SIR}_k\geq \tau\bigg|\bigcap_{i=1,2,\ldots,k-1}\text{SIR}_i\geq \tau\right]\right)\text{Pr} \left[\bigcap_{i=1,2,\ldots,k-1}\text{SIR}_i\geq \tau\right]\nonumber \\
    &+ \text{Pr} \Bigg[\text{SIR}_{\lceil\frac{n}{m_j}\rceil}\geq \tau\bigg|
    \bigcap_{i=1,2,\ldots,\lceil\frac{n}{m_j}\rceil-1}\text{SIR}_i\geq \tau\Bigg]
    \text{Pr} \Bigg[\bigcap_{i=1,2,\ldots,\lceil\frac{n}{m_j}\rceil-1}\text{SIR}_i\geq \tau\Bigg] \nonumber\\
    =&\sum\limits_{k=1}^{\lceil\frac{n}{m_j}\rceil} \frac{(k-1)m_j}{n}\left(1-q_k(\tau)\right)\prod\limits_{i=1}^{k-1}q_i(\tau) + \prod\limits_{i=1}^{\lceil\frac{n}{m_j}\rceil}q_i(\tau).
\end{align}

Substituting \eqref{eqn:CodedTransmissionCoverageProbability} from Lemma~\ref{Lem:CodedTransmissionCoverageProbability} into \eqref{eqn:R_c-sum-1} and letting $Q_{\tau} = 1+\frac2{\alpha}\tau^{\frac2{\alpha}}B'\left(\frac2{\alpha},1-\frac2{\alpha},\frac1{1+\tau}\right)$, we have
\begin{align}\label{eqn:Prove-R_c-2}
    L[m_j] =& \sum\limits_{k=1}^{\lceil\frac{n}{m_j}\rceil} \frac{(k-1)m_j}n \left(1-Q_{\tau}^{-k}\right) Q_{\tau}^{-\frac{k(k-1)}2} + Q_{\tau}^{-\frac{\bigl\lceil\frac{n}{m_j}\bigr\rceil\left(\bigl\lceil\frac{n}{m_j}\bigr\rceil+1\right)}2}\nonumber \\
    =&\frac{m_j}n \sum\limits_{k=1}^{\lceil\frac{n}{m_j}\rceil}(k-1) \left(Q_{\tau}^{-\frac{k(k-1)}2}-Q_{\tau}^{-\frac{k(k+1)}2}\right)  + Q_{\tau}^{-\frac{\bigl\lceil\frac{n}{m_j}\bigr\rceil\left(\bigl\lceil\frac{n}{m_j}\bigr\rceil+1\right)}2}\nonumber \nonumber \\
    =&\frac{m_j}n\sum\limits_{k=1}^{\lceil\frac{n}{m_j}\rceil} Q_{\tau}^{-\frac{k(k+1)}2} +
    \left(1-\frac{m_j}{n}\biggl\lceil\frac{n}{m_j}\biggr\rceil\right)
    Q_{\tau}^{-\frac{\bigl\lceil\frac{n}{m_j}\bigr\rceil\left(\bigl\lceil\frac{n}{m_j}\bigr\rceil+1\right)}2}\nonumber \nonumber \\
    =&\frac{m_j}n\sum\limits_{k=1}^{\lceil\frac{n}{m_j}\rceil}C_k(\tau)
    +\left(1-\frac{m_j}{n}\biggl\lceil\frac{n}{m_j}\biggr\rceil\right)
    C_{\lceil\frac{n}{m_j}\rceil}(\tau).
\end{align}
Thus, Lemma~\ref{Lem:BRofCodedCaching} is proved.

\section*{Appendix C: Proof of Lemma~\ref{Lem:R-xj-increase-concave}}
When $x_j \in \left[\frac1{t_x},\frac1{t_x-1}\right)$, $t_x\in \{2,3,\cdots\}$, we have $\bigl\lceil\frac1{x_j}\bigr\rceil=t_x$.
Hence $L(x_j)$ is a piece-wise linear function.

In each region $x_j \in \left(\frac1{t_x},\frac1{t_x-1}\right)$,
the derivative of $L(x_j)$ can be obtained as
\begin{align}\label{eqn:derivative-R-c}
 \frac{\mathrm{d}L(x_j)}{\mathrm{d}x_j} = \left(\sum\limits_{i=1}^{t_x} C_i(\tau)-t_x  C_{t_x}(\tau)\right) > 0.
\end{align}
In \eqref{eqn:derivative-R-c}, the last inequality holds due to the decreasing property of $C_k(\tau)$. Thus, $L(x_j)$ is an increasing function.

We now prove the concavity of $L(x_j)$.
Assuming $x_j \in \left(\frac1{t_x},\frac1{t_x-1}\right)$ and $x'_j\in \left(\frac1{t_x+1},\frac1{t_x}\right)$, we have that
\begin{align}\label{eqn:thedifference}
 &\frac{\mathrm{d}L(x_j)}{\mathrm{d}x_j}-\frac{\mathrm{d}L(x'_j)}{\mathrm{d}x'_j}
  = t_x\left(C_{t_x+1}(\tau)-C_{t_x}(\tau)\right) < 0.
\end{align}
Equation~\eqref{eqn:thedifference} means that the successive slopes of $L(x_j)$ are decreasing.
$L(x_j)$ is thus a concave function. Lemma~\ref{Lem:R-xj-increase-concave} is thus proved.


\section*{Appendix D: Proof of Lemma~\ref{Lem:ErgodicRateofCodedCaching}}
 When $m_j=0$, there is no packet of $W_j$ in each SBS, thus the ergodic rate is $0$.

 When $m_j \in \mathcal{N^+}$, we have
\begin{align}\label{}
  R[m_j] =& \biggl\lceil\frac{n}{m_j}\biggr\rceil \int_0^{\infty} \text{Pr} \left[\log\left(1+\min\left(\text{SIR}_1,\ldots,\text{SIR}_{\bigl\lceil\frac{n}{m_j}\bigr\rceil}\right)\right) \geq r\right] \mathrm{d}r \nonumber \\
  =& \biggl\lceil\frac{n}{m_j}\biggr\rceil \int_0^{\infty} \text{Pr} \left[\min\left(\text{SIR}_1,\ldots,\text{SIR}_{\bigl\lceil\frac{n}{m_j}\bigr\rceil}\right) \geq 2^r - 1\right] \mathrm{d}r \nonumber \\
  =& \biggl\lceil\frac{n}{m_j}\biggr\rceil \int_0^{\infty} \text{Pr} \left[\text{SIR}_1\geq 2^r - 1, \ldots,  \text{SIR}_{\bigl\lceil\frac{n}{m_j}\bigr\rceil}\geq 2^r - 1\right] \mathrm{d}r \nonumber \\
  \overset{(a)}{\approx}& \biggl\lceil\frac{n}{m_j}\biggr\rceil \int_0^{\infty}\prod\limits_{k=1}^{\lceil\frac{n}{m_j}\rceil} \text{Pr} \left[\text{SIR}_k\geq 2^r - 1\right] \mathrm{d}r \nonumber \\
  \overset{(b)}{=}& \biggl\lceil\frac{n}{m_j}\biggr\rceil \int_0^{\infty} C_{\lceil\frac{n}{m_j}\rceil}(2^r-1) \mathrm{d}r.
\end{align}
Here, $(a)$ is due to the independence assumption of $\{\text{SIR}_k\}$ and $(b)$ follows from~\eqref{eqn:q-k-proof} and by letting $\tau = 2^r-1$ in~\eqref{eqn:c_k_tau}.

\section*{Appendix E: Proof of Theorem~\ref{Thm:MorePopularMorePackets}}
We assume that there are two files in $\mathcal{W}$ and file $W_1$ is more popular than file $W_2$, i.e., $p_1>p_2$. However, we assume that less packets of file $W_1$ is stored in each SBS than that of file $W_2$, i.e., $m_1<m_2$, which has maximized the AFOT performance. We have the sum of FOT for file $W_1$ and file $W_2$ as
\begin{align}\label{}
    L_{12}=p_1 L[m_1] + p_2 L[m_2].
\end{align}

Now, we try to store $m_2$ packets of file $W_1$ and $m_1$ packets of file $W_2$. Thus, the sum of FOT for file $W_1$ and file $W_2$ can be calculated as
\begin{align}\label{}
   L_{21}=p_1 L[m_2] + p_2 L[m_1].
\end{align}
Then, we compare $L_{12}$ with $L_{21}$
\begin{align}\label{}
    L_{12}-L_{21}=(p_1-p_2) \left(L[m_1]-L[m_2]\right).
\end{align}
We have known that $L[m_j]$ is an increasing sequence of $m_j$. Therefore, $L[m_1]<L[m_2]$. Then, we can get $L_{12}<L_{21}$. That means if we store $m_2$ packets of file $W_1$ and $m_1$ packets of file $W_2$, AFOT can be larger. Therefore, more popular files store more packets in cache.


\section*{Appendix F: Proof of Corollary~\ref{Cor:tau-gamma-MPC}}
Let $x=\left(1+\frac2{\alpha}\tau^{\frac2{\alpha}} B'\left(\frac2{\alpha},1-\frac2{\alpha},\frac1{1+\tau}\right)\right)^{-1}$.
For $\tau > 0$, $x \in (0,1)$ is a decreasing function of $\tau$.
Define $g(x)=\frac{\sum_{k=1}^nC_k(\tau)}{C_1(\tau)-C_2(\tau)}$.
We thus have
\begin{align}\label{}
  \frac{\mathrm{d}{g(x)}}{\mathrm{d}x} = \frac{\sum\limits_{k=1}^n x^{\frac{k(k+1)}2}\left(\frac{k(k+1)}2-1-x^2 \left(\frac{k(k+1)}2-3\right)\right)}{x^2\left(1-x^2\right)^2}.
\end{align}
Since $x\in (0,1)$, $\frac{\mathrm{d}{g(x)}}{\mathrm{d}x}>0$.
Thus, $g(x)$ is an increasing function of $x$.
Overall, $\frac{\sum_{k=1}^nC_k(\tau)}{C_1(\tau)-C_2(\tau)} \in (1,+\infty)$ is a decreasing function of $\tau$.
If $p_M > p_{M+1}$, there must exist a $\tau_0>0$ such that inequality~\eqref{eqn:condition-degenrate-MPC} satisfies when $\tau \geq \tau_0$.
The first part of the Corollary is thus proven.

Furthermore, if the file popularity follows the Zipf distribution with shape parameter $\gamma >0$, i.e., $p_j = j^{-\gamma}/\sum_{f=1}^F f^{-\gamma}, j\in\mathcal{F}.$
We then have
\begin{align}\label{eqn:gamma-MPC-proof}
  \frac{p_M}{p_{M+1}} = \left(1+\frac1M\right)^{\gamma} >1.
\end{align}
Equation~\eqref{eqn:gamma-MPC-proof} is an increasing function of $\gamma$.
Hence, for any given $\tau>0$, there must exist a $\gamma_0>0$ such that inequality~\eqref{eqn:condition-degenrate-MPC} satisfies when $\gamma \geq \gamma_0$.
The second part of the Corollary is also proven.


\bibliographystyle{IEEEtran}
\bibliography{IEEEabrv,ICC_reference_v2}

\end{document}